\documentclass{scrartcl}

\usepackage{cmap}
\usepackage[utf8]{inputenc}
\usepackage[T1]{fontenc}

\usepackage{amsmath}
\usepackage{amssymb}
\usepackage{amsthm}
\usepackage{thmtools}
\usepackage{bm}
\usepackage{mathrsfs}
\usepackage{csquotes}
\usepackage{tabularx}
\usepackage{enumitem}
\usepackage{algorithm}
\usepackage{algpseudocode}
\usepackage{bbm}
\usepackage{textcomp}
\usepackage{varwidth}
\usepackage[bottom]{footmisc}
\usepackage{tikz}
\usetikzlibrary{automata, calc, matrix, shapes, positioning, decorations.pathreplacing, decorations.pathmorphing, fit}

\usepackage{todonotes}

\usepackage{hyperref}

\setcapindent{0pt}

\theoremstyle{plain}
\newtheorem{thm}{Theorem}[section]
\newtheorem{lemma}[thm]{Lemma}
\newtheorem{prop}[thm]{Proposition}
\newtheorem{cor}[thm]{Corollary}
\newtheorem{fact}[thm]{Fact}

\theoremstyle{definition}

\newtheorem{example}[thm]{Example}
\newtheorem{prob}[thm]{Open Problem}
\newtheorem{counterexample}[thm]{Counter-Example}

\theoremstyle{remark}
\newtheorem{remark}[thm]{Remark}

\usepackage{xspace}

\newcommand*{\SAut}{$\mathscr{S}$\kern-0.4ex-au\-to\-ma\-ton\xspace}
\newcommand*{\SAuta}{$\mathscr{S}$\kern-0.4ex-au\-to\-ma\-ta\xspace}
\newcommand*{\GAut}{$\mathscr{G}$\kern-0.2ex-au\-to\-ma\-ton\xspace}
\newcommand*{\GAuta}{$\mathscr{G}$\kern-0.2ex-au\-to\-ma\-ta\xspace}

\DeclareFontFamily{U}{mathb}{\hyphenchar\font45}
\DeclareFontShape{U}{mathb}{m}{n}{
  <5> <6> <7> <8> <9> <10> gen * mathb
  <10.95> mathb10 <12> <14.4> <17.28> <20.74> <24.88> mathb12
}{}
\DeclareSymbolFont{mathb}{U}{mathb}{m}{n}
\DeclareFontSubstitution{U}{mathb}{m}{n}
\DeclareMathSymbol{\drsh}{3}{mathb}{"EB}

\newlength{\edgelength}
\newcommand{\trans}[4]{%
  \begin{tikzpicture}[auto, shorten >=1pt, >=latex, baseline=(l.base), inner sep=0pt, outer xsep=0.3333em]
    \node (l) {\ensuremath{#1}};%
    \setlength{\edgelength}{\widthof{\scriptsize\ensuremath{#2/#3}}+0.5cm}%
    \node[base right=\edgelength of l] (r) {\ensuremath{#4}};%
    \path[->] (l.mid east) edge node[inner sep=0pt] {\scriptsize\ensuremath{#2/#3}} (r.mid west);%
  \end{tikzpicture}%
}

\DeclareMathOperator{\id}{id}

\DeclareMathOperator{\Pre}{Pre}
\DeclareMathOperator{\Suf}{Suf}
\newcommand*{\rev}[1]{\partial #1}

\newcommand{\problem}[3][]{%
  \par\vspace{0.125cm plus 0.05cm minus 0.05cm}\begin{tabularx}{\textwidth-2\parindent}{lX}%
    \if\relax\detokenize{#1}\relax%
    \else%
      \textnormal{\textbf{Constant:}}&#1\\%
    \fi%
    \textnormal{\textbf{Input:}}&#2\\%
    \textnormal{\textbf{Question:}}&#3\\%
  \end{tabularx}\vspace{0.125cm plus 0.05cm minus 0.05cm}\par%
  }

\usepackage[affil-it]{authblk}

\author{Daniele~D'Angeli\thanks{The first author was supported by the Austrian Science Fund project FWF P29355-N35.}}
\affil{Università degli Studi Niccolò Cusano\\
  Via Don Carlo Gnocchi, 3\\
  00166 Roma, Italy}
\author{Dominik~Francoeur\thanks{The second author was supported by a Doc.Mobility grant from the Swiss National Science Foundation as well as the "@raction" grant ANR-14-ACHN-0018-01 during a visit at the École Normale Supérieure in Paris.}}
\affil{Unité de mathématiques pures et appliquées\\
  ENS de Lyon\\
  46, allée d’Italie, 69364 Lyon Cedex 07, France}
\author{Emanuele~Rodaro}
\affil{Department of Mathematics\\
  Politecnico di Milano\\
  Piazza Leonardo da Vinci, 32\\
  20133 Milano, Italy}
\author{Jan~Philipp~Wächter}
\affil{Institut für Formale Methoden der Informatik (FMI)\\
  Universität Stuttgart\\
  Universitätsstraße 38\\
  70569 Stuttgart, Germany}

\title{On the Orbits of Automaton Semigroups\\and Groups}

\begin{document}
  \maketitle

  \begin{abstract}
    We investigate the orbits of automaton semigroups and groups to obtain algorithmic and structural results, both for general automata but also for some special subclasses.
    
    First, we show that a more general version of the finiteness problem for automaton groups is undecidable. This problem is equivalent to the finiteness problem for left principal ideals in automaton semigroups generated by complete and reversible automata.
    
    Then, we look at $\omega$-word (i.\,e.\ right infinite words) with a finite orbit. We show that every automaton yielding an $\omega$-word with a finite orbit already yields an ultimately periodic one, which is not periodic in general, however. On the algorithmic side, we observe that it is not possible to decide whether a given periodic $\omega$-word has an infinite orbit and that we cannot check whether a given reversible and complete automaton admits an $\omega$-word with a finite orbit, a reciprocal problem to the finiteness problem for automaton semigroups in the reversible case.
    
    Finally, we look at automaton groups generated by reversible but not bi-reversible automata and show that many words have infinite orbits under the action of such automata.
  \end{abstract}

  \begin{section}{Introduction}\enlargethispage{\baselineskip}
    Starting with the realization that the famous Grigorchuk group\footnote{See \cite{grigorchuk2008groups} for an accessible introduction to Grigorchuk's group.} (the first example of a group with subexponential but superpolynomial growth) and many other groups with interesting and peculiar properties can be generated by finite automata, the class of so-called automaton groups grew into a widely and intensively studied object. The finite automaton here is a finite-state, letter-to-letter transducer. It induces an action of the finite words over its state set on the words over its alphabet and this action is used to define the generated group.
    
    Roughly speaking, the research in this area is divided into three branches: the study of individual, special automaton groups (such as the mentioned Grigorchuk group), the study of structural properties of automaton groups and the study of algorithmic problems over automaton groups. The aim of this work is to contribute to the latter two of these branches. In fact, we will not only consider automaton groups but also their natural generalization to automaton semigroups, in which the interest seems to have risen lately, both for structural results (see, for example, \cite{cain2009automaton, brough2015automaton, brough2017automatonTCS, structurePart}) but also for algorithmic problems. They are particularly interesting for the latter point because many important and classical algorithmic problem in group theory are proven\footnote{Two of Dehn's fundamental problems in algorithmic group theory, the conjugacy problem and the isomorphism problems are among the problems which have been proven to be undecidable for automaton groups \cite{Su-Ve09}.} or suspected to be undecidable for automaton groups and it is usually easier to encode computations in semigroups than in groups. Sometimes algorithmic results for automaton semigroups could later be lifted to groups. An example for this is the order problem: first, it could be shown to be undecidable for automaton semigroups \cite{Gilbert13} and later this result could be extended to automaton groups \cite{bartholdi2017wordAndOrderProblems, gillibert2018automaton}. Similarly, a result on the complexity of the word problem could be lifted from the (inverse) semigroup case \cite{DAngeli2017} to the groups case \cite{pspacePart}. For another important problem, the finiteness problem, the current state is that it has been proven to be undecidable in the semigroup case \cite{Gilbert13} (and also in a more restrictive setting \cite{decidabilityPart}) but the decidability of the problem in the group case remains unknown.
    
    In this paper, we give a partial solution to this problem. The classical question of the finiteness problem is whether a given invertible automaton generates a finite or an infinite group. This is equivalent to the question whether there are infinitely many state sequences whose actions on the words over the alphabet are pairwise distinct. We show that the problem is undecidable if we instead ask whether there are infinitely many state sequences whose actions pairwisely differ on all words with a given prefix. If we pass to the dual automaton (i.\,e.\ if we change the roles of states and letters), this problem is the same as asking whether a given element $s$ of a semigroup $S$ generated by a complete and reversible (i.\,e.\ co-deterministic with respect to the input) automaton has an infinite left principal ideal $Ss \cup \{ s \}$.
    
    Quite recently, the current authors could show that the algebraic property of an automaton semigroup (and, thus, of an automaton group) to be infinite is equivalent to the fact that the action given by the generating automaton yields an $\omega$-word with an infinite orbit \cite{orbitsPart}. Thus, the finiteness problem is equivalent to asking whether such an infinite orbit exists. Reciprocally, we can also ask whether there is an (infinite) $\omega$-word with a finite orbit and we show that the corresponding decision problem is undecidable even for complete and reversible automata. Furthermore, we show that it is algorithmically impossible to test whether a given (periodic) word has a finite or infinite orbit.
    
    Structurally, we explore further consequences of the mentioned connection between the semigroup being infinite and the existence of an infinite orbit as well as the dual argument underlying its proof. Here, we first look at $\omega$-words with a finite orbit and show that, whenever such a word exists, there is also an ultimately periodic $\omega$-word with a finite orbit. We will see that this word can be assumed to be periodic if the automaton is reversible but that this does not hold in the general case. Finally, we look at the class of groups generated by invertible, reversible but not bi-reversible (i.\,e.\ not co-deterministic with respect to the output) automata. Here, we obtain that they always admit periodic $\omega$-words (of a certain form) with infinite orbits and that, if the dual automaton is additionally connected, all $\omega$-words have infinite orbits. For semigroups generated by reversible but not bi-reversible automata, we will see, however, that this is not true: they can be infinite while the orbits of all (ultimately) periodic words are finite. This also shows that the result about the existence of a word with an infinite orbit (if the generated semigroup is infinite) cannot be extended to periodic or ultimately periodic words.
  \end{section}

  \begin{section}{Preliminaries}\label{sec:preliminaries}
    \paragraph{Fundamentals, Words and Languages.}
    Let $A$ and $B$ be sets. We write $A \sqcup B$ for their disjoint union and, for a \emph{partial function} from $A$ to $B$, we write $A \to_p B$. If the function is total, we omit the index $p$. Furthermore, we use $\mathbb{N}$ to denote the set of natural numbers including $0$.
    
    A non-empty, finite set $\Sigma$ is called an \emph{alphabet}, its elements are called \emph{letters} and finite or right-infinite sequence over $\Sigma$ are called \emph{finite words} and \emph{$\omega$-words}, respectively. A \emph{word} can be a finite word or an $\omega$-word. The set of all finite word over $\Sigma$ -- including the empty word $\varepsilon$ -- is $\Sigma^*$ and $\Sigma^+$ is $\Sigma^* \setminus \{ \varepsilon \}$. The length of a finite word $w = a_1 \dots a_{\ell}$ with $a_1, \dots, a_{\ell} \in \Sigma$ is $|w| = \ell$. Finally, the \emph{reverse} of a finite word $w = a_1 \dots a_{\ell}$ with $a_1, \dots, a_{\ell} \in \Sigma$ is $\rev{w} = a_{\ell} \dots a_1$.
    \enlargethispage{1\baselineskip}
    
    The set of all $\omega$-words over an alphabet $\Sigma$ is $\Sigma^\omega$. An $\omega$-word is called \emph{ultimately periodic} if it is of the form $uv^\omega$ for some $u \in \Sigma^*$ and $v \in \Sigma^+$ where $v^\omega = v v \dots$; it is called \emph{periodic} if it is of the form $v^\omega$ for $v \in \Sigma^+$. We can also take the reverse of an $\omega$-word $\alpha = a_1 a_2 \dots$ with $a_1, a_2, \dots \in \Sigma$ to obtain the \emph{left-infinite sequence} $\rev{\alpha} = \dots a_2 a_1$ over $\Sigma$.
    
    A word $u$ is called a \emph{suffix} of another word $w$ if there is some finite word $x$ with $w = xu$. Symmetrically, $u$ is a \emph{prefix} of $w$ if there is a word $x$ with $w = ux$. A \emph{language} $L$ is a set of words over some alphabet $\Sigma$. It is \emph{suffix-closed} if $w \in L$ implies that every suffix of $w$ is in $L$ as well and it is \emph{prefix-closed} if $w \in L$ implies that every prefix of $w$ is also in $L$. By $\Pre w$, we denote the set of finite prefixes of a word $w$ and $\Pre L$ for a language $L$ is $\Pre L = \cup_{w \in L} \Pre w$. Symmetrically, we define $\Suf w$ and $\Suf L$ for the finite suffixes of a finite word or left-infinite sequence $w$ and a set $L$ of finite words and left-infinite sequences.
    
    For two languages $K$ and $L$ of finite words, we let $KL = \{ uv \mid u \in K, v \in L \}$. Furthermore, we define $L^* = \{ w_1 \dots w_i \mid i \in \mathbb{N}, w_1, \dots, w_i \in L \}$ and sometimes simply write $w$ for the singleton language $\{ w \}$. Additionally, we lift operators on words to languages; for example, we let $\partial L = \{ \partial w \mid w \in L \}$.
    
    \paragraph{Semigroups, Groups and Torsion.}
    We assume the reader to be familiar with basic notions from semigroup and group theory such as inverses (in the group sense) and generating sets. If a semigroup $S$ or a monoid $M$ is generated by a (finite) set $Q$, then there is a natural epimorphism from $Q^+$ to $S$ or from $Q^*$ to $M$. In this case, we write $\bm{q}$ \emph{in $S$} or \emph{in $M$} for the image of $\bm{q} \in Q^+$ or $\bm{q} \in Q^*$. Similarly, for $K \subseteq Q^*$, we write $K$ \emph{in $S$} or \emph{in $M$} for the image of $K$ under this homomorphism. Additionally, we use some natural variations for this notation. For example, we write $\bm{p} = \bm{q}$ in $S$ if $\bm{p}$ and $\bm{q}$ have the same image under the natural homomorphism.
    
    An element $s$ of a semigroup $S$ has \emph{torsion} if there are $i, j \geq 1$ with $i \neq j$ but $s^i = s^j$. If $S = G$ is even a group, this is connected to the \emph{order} of a group element $g \in G$: it is the smallest number $i \geq 1$ such that $g^i$ is the neutral element of the group; if there is no such $i$, then the element has infinite order. Obviously, an element of a group is of finite order if and only if it has torsion.
    
    \paragraph{Automata.}
    The most important objects in this paper are automata, which are more precisely described as finite-state, letter-to-letter transducers. Formally, an \emph{automaton} is a triple $\mathcal{T} = (Q, \Sigma, \delta)$ where $Q$ is a set of \emph{states}, $\Sigma$ is an alphabet and $\delta \subseteq Q \times \Sigma \times \Sigma \times Q$ is a set of \emph{transitions}. For a transition $(p, a, b, q) \in Q \times \Sigma \times \Sigma \times Q$, we use a more graphical notation and denote it by $\trans{p}{a}{b}{q}$ or, when depicting an entire automaton, by
    
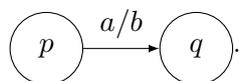
\begin{center}
      \begin{tikzpicture}[baseline=(p.base), auto, >=latex]
        \node[state] (q) {$p$};
        \node[state, right=of q] (p) {$q$};
        \draw[->] (q) edge node {$a / b$} (p);
      \end{tikzpicture}.
    \end{center}
    
    An automaton $\mathcal{T} = (Q, \Sigma, \delta)$ is \emph{complete} if
    \[
      d_{p, a} = \left| \{ \trans{p}{a}{b}{q} \in \delta \mid b \in \Sigma, q \in Q \} \right|
    \]
    is at least one for every $p \in Q$ and $a \in \Sigma$. If, on the other hand, all $d_{p, a}$ are at most one, then $\mathcal{T}$ is \emph{deterministic}. Additionally, $\mathcal{T}$ is \emph{reversible} if it is co-deterministic with respect to the input, i.\,e.\ if
    \[
      \{ \trans{p}{a}{b}{q} \in \delta \mid p \in Q, b \in \Sigma \}
    \]
    contains at most one element for every $a \in \Sigma$ and $q \in Q$ and it is \emph{inverse-reversible} if it is co-deterministic with respect to the output, i.\,e.\ if
    \[
      \{ \trans{p}{a}{b}{q} \in \delta \mid p \in Q, a \in \Sigma \}
    \]
    contains at most one element for every $b \in \Sigma$ and $q \in Q$. An automaton that is both, reversible and inverse-reversible is called \emph{bi-reversible}.
    
    Another way of depicting transitions in automata are \emph{cross diagrams}. We write
    \begin{center}
      \begin{tikzpicture}[baseline=(m-3-2.base)]
        \matrix[matrix of math nodes, text height=1.25ex, text depth=0.25ex] (m) {
            & a & \\
          p & & q \\
            & b & \\
        };
        \foreach \i in {1} {
          \draw[->] let
            \n1 = {int(2+\i)}
          in
            (m-2-\i) -> (m-2-\n1);
          \draw[->] let
            \n1 = {int(1+\i)}
          in
            (m-1-\n1) -> (m-3-\n1);
        };
      \end{tikzpicture}.
    \end{center}
    to indicate that an automaton $\mathcal{T} = (Q, \Sigma, \delta)$ contains the transition $\trans{p}{a}{b}{q} \in \delta$. We can combine multiple transitions into a single cross diagrams. For example, the cross diagram
    \begin{center}
      \begin{tikzpicture}
        \matrix[matrix of math nodes, text height=1.25ex, text depth=0.25ex] (m) {
                   & a_{0, 1}     &          & \dots &              & a_{0, m}     &     \\
          q_{1, 0} &              & q_{1, 1} & \dots & q_{1, m - 1} &              & q_{1, m} \\
                   & a_{1, 1}     &          &       &              & a_{1, m}     &     \\
            \vdots & \vdots       &          &       &              & \vdots       & \vdots \\
                   & a_{n - 1, 1} &          &       &              & a_{n - 1, m} &     \\
          q_{n, 0} &              & q_{n, 1} & \dots & q_{n, m - 1} &              & q_{n, m} \\
                   & a_{n, 1}     &          & \dots &              & a_{n, m}     &     \\
        };
        \foreach \j in {1, 5} {
          \foreach \i in {1, 5} {
            \draw[->] let
              \n1 = {int(2+\i)},
              \n2 = {int(1+\j)}
            in
              (m-\n2-\i) -> (m-\n2-\n1);
            \draw[->] let
              \n1 = {int(1+\i)},
              \n2 = {int(2+\j)}
            in
              (m-\j-\n1) -> (m-\n2-\n1);
          };
        };
      \end{tikzpicture}
    \end{center}
    states that the automaton contains all transitions $\trans{q_{i, j - 1}}{a_{i - 1, j}}{a_{i, j}}{q_{i, j}}$ for $1 \leq i \leq n$ and $1 \leq j \leq m$. Sometimes, we will omit intermediate states or letters if we do not need to assign them a name. Instead of always drawing complete cross diagrams, we also introduce a shot-hand notation where we do not only allow states and letters but also state sequences and words. For example, the above cross diagram can be abbreviate by
    \begin{center}
      \begin{tikzpicture}
        \matrix[matrix of math nodes, text height=1.25ex, text depth=0.25ex] (m) {
                                           & u = a_{0, 1} \dots a_{0, m} & \\
          q_{n, 0} \dots q_{1, 0} = \bm{q} &                             & \bm{p} = q_{n, m} \dots q_{1, m} \\
                                           & v = a_{n, 1} \dots a_{n, m} & \\
        };
        \foreach \i in {1} {
          \draw[->] let
            \n1 = {int(2+\i)}
          in
            (m-2-\i) -> (m-2-\n1);
          \draw[->] let
            \n1 = {int(1+\i)}
          in
            (m-1-\n1) -> (m-3-\n1);
        };
      \end{tikzpicture}.
    \end{center}
    It is important to note the ordering of the state sequences here: $q_{n, 0}$ belongs to the last transition but is written leftmost while $q_{1, 0}$ belongs to the first transition and is written rightmost.\footnote{This seemingly wrong ordering is justified here because we will define automaton semigroups and groups using left actions later on.}
    
    \paragraph{Automaton Semigroups.}
    For a deterministic automaton $\mathcal{T} = (Q, \Sigma, \delta)$, we can define a partial left action of $Q^*$ on $\Sigma^*$ and a partial right action of $\Sigma^*$ on $Q^*$ using cross diagrams. Since the automaton is deterministic, there is at most one $v \in \Sigma^+$ and at most one $\bm{q} \in Q^+$ for every $u \in \Sigma^+$ and every $\bm{p} \in Q^+$ such that the cross diagram
    \begin{center}
      \begin{tikzpicture}[baseline=(m-3-2.base)]
        \matrix[matrix of math nodes, text height=1.25ex, text depth=0.25ex] (m) {
                 & u & \\
          \bm{p} &   & \bm{q} \\
                 & v & \\
        };
        \foreach \i in {1} {
          \draw[->] let
            \n1 = {int(2+\i)}
          in
            (m-2-\i) -> (m-2-\n1);
          \draw[->] let
            \n1 = {int(1+\i)}
          in
            (m-1-\n1) -> (m-3-\n1);
        };
      \end{tikzpicture}
    \end{center}
    holds. In this case, we define the left partial action of $\bm{p}$ on $u$ as $\bm{p} \circ_{\mathcal{T}} u = v$ and the right partial action of $u$ on $\bm{p}$ as $\bm{p} \cdot_{\mathcal{T}} u = \bm{q}$. If there are no such $\bm{q}$ and $v$, we let $\bm{p} \circ_{\mathcal{T}} u$ and $\bm{p} \cdot_{\mathcal{T}} u$ be undefined. Additionally, we let $\bm{p} \circ_{\mathcal{T}} \varepsilon = \varepsilon$, $\varepsilon \circ_{\mathcal{T}} u = u$, $\bm{p} \cdot_{\mathcal{T}} \varepsilon = \bm{p}$ and $\varepsilon \cdot_{\mathcal{T}} \bm{p} = \varepsilon$. With this definition, it is easy to see that we have $\bm{q} \circ_{\mathcal{T}} \bm{p} \circ_{\mathcal{T}} u = \bm{qp} \circ_{\mathcal{T}} u$ and $u \cdot_{\mathcal{T}} \bm{p} \cdot \bm{q} = u \cdot_{\mathcal{T}} \bm{pq}$. Whenever the automaton $\mathcal{T}$ is clear form the context, we simply write $\bm{p} \circ u$ and $\bm{p} \cdot u$ instead of $\bm{p} \circ_{\mathcal{T}} u$ and $\bm{p} \cdot_{\mathcal{T}} u$.
    
    Now, every $\bm{p} \in Q^*$ induces a partial, length-preserving function $\bm{p} \circ{}\!: \Sigma^* \to_p \Sigma^*$ which maps every $u$ to $\bm{p} \circ u$. These partial functions are prefix-compatible in the sense that we have $\bm{p} \circ u_1 u_2 = (\bm{p} \circ u_1) v_2$ for some $v_2 \in \Sigma^*$ (whenever the partial action is defined on a word $u_1 u_2$). Naturally, we can extend $\bm{p} \circ{}\!$ into a partial function $\Sigma^* \cup \Sigma^\omega \to_p \Sigma^* \cup \Sigma^\omega$: $\alpha = a_1 a_2 \dots$ with $a_1, a_2, \dots \in \Sigma$ gets mapped to $b_1 b_2 \dots$ where the $b_1, b_2, \dots \in \Sigma$ are defined by $b_1 \dots b_m = \bm{p} \circ a_1 \dots a_m$ (if $\bm{p} \circ{}\!$ is undefined on some prefix of $\alpha$, then $\bm{p} \circ{}\!$ obviously should also be undefined on $\alpha$).
    
    In the same way, we can also define the partial, length-preserving functions ${}\! \cdot u: Q^* \to_p Q^*$ with $u \in \Sigma^*$ which map $\bm{p}$ to $\bm{p} \cdot u$ and observe that they have similar properties as the maps $\bm{p} \circ{}\!$.
    
    The \emph{semigroup $\mathscr{S}(\mathcal{T})$ generated} by the deterministic automaton $\mathcal{T}$ is the set $Q^+ \circ{}\! = \{ \bm{q} \circ{}\! \mid \bm{q} \in Q^+ \}$ with the composition of partial functions as its operation. This semigroup is generated by $Q \circ{}\! = \{ q \circ{}\! \mid q \in Q \}$. To emphasize the fact, that they generate semigroups, we will use the name \SAuta for deterministic automata from now on. An \emph{automaton semigroup} is a semigroup generated by some \SAut.
    
    \begin{remark}
      We want to point out that we do not require an \SAut to be complete. If an automaton semigroup is generated by a complete automaton, we call it a \emph{complete automaton semigroup} to emphasize this. It is not unknown whether the class of complete automaton semigroups and the class of (partial) automaton semigroups coincide (see \cite{structurePart} for a discussion).
      
      Clearly, if $\mathcal{T} = (Q, \Sigma, \delta)$ is a complete \SAut, $\bm{p} \circ u$ and $\bm{p} \cdot u$ are defined for all $\bm{p} \in Q^*$ and all $u \in \Sigma^*$ and all functions $\bm{q} \circ{}\!$ are total.
    \end{remark}
    
    For an \SAut $\mathcal{T} = (Q, \Sigma, \delta)$, the partial action of $\Sigma^*$ on $Q^*$ is compatible with the structure of the generated semigroup as we have $\bm{p} \circ{}\! = \bm{q} \circ{}\! \implies \bm{p} \cdot u \circ{}\! = \bm{q} \cdot u \circ{}\!$ (or both undefined) for all $\bm{p}, \bm{q} \in Q^*$ and $u \in \Sigma^*$ (which can be seen easily). Accordingly, we can define a partial action of $\Sigma^*$ on $\mathscr{S}(\mathcal{T})$: for an element $s = \bm{q} \circ{}\! \in \mathscr{S}(\mathcal{T})$ with $\bm{q} \in Q^+$, we let $s \cdot u = \bm{p} \cdot u \circ{}\!$ for $u \in \Sigma^*$.
    
    \paragraph{Automaton Groups and Inverse Automata.}\enlargethispage{1\baselineskip}
    An automaton $\mathcal{T} = (Q, \Sigma, \delta)$ is called \emph{invertible} if the sets
    \[
      \{ \trans{p}{a}{b}{q} \mid a \in \Sigma, q \in Q \}
    \]
    contain at most one element for all $p \in Q$ and $b \in \Sigma$. If a complete \SAut $\mathcal{T} = (Q, \Sigma, \delta)$ is invertible, all functions $\bm{p} \circ{}\!$ with $\bm{p} \in Q^*$ are bijections (and, in particular, total). In this case, we define $Q^{-1} = \{ q^{-1} \mid q \in Q \}$ as a disjoint copy of $Q$ and let $Q^{\pm*} = (Q \sqcup Q^{-1})^*$. We can extend the action of $Q^*$ on $\Sigma^*$ (and $\Sigma^\omega$) into an action of $Q^{\pm*}$ on $\Sigma^*$ (and $\Sigma^\omega$) by letting $q^{-1} \circ u$ be given by the pre-image of $u$ under $q^{-1} \circ{}\!$.
    
    The \emph{group $\mathscr{G}(\mathcal{T})$ generated} by an complete and invertible \SAut $\mathcal{T} = (Q, \Sigma, \delta)$ is $Q^{\pm *} \circ{}\! = \{ \bm{q} \circ{}\! \mid \bm{q} \in Q^{\pm *} \}$ with the composition of functions as its operation and such an automaton is called a \GAut. A group generated by some \GAut is an \emph{automaton group}.
    
    The group $\mathscr{G}(\mathcal{T})$ generated by a \GAut $\mathcal{T} = (Q, \Sigma, \delta)$ is also an automaton semigroup. It is the semigroup generated by the automaton $\mathcal{T}' = (Q \sqcup Q^{-1}, \Sigma, \delta \cup \delta^{-1})$ where we let
    \[
      \delta^{-1} = \{ \trans{p^{-1}}{b}{a}{q^{-1}} \mid \trans{p}{a}{b}{q} \in \delta \} \text{.}
    \]
    The automaton $\mathcal{T}^{-1} = (Q^{-1}, \Sigma, \delta^{-1})$ is the \emph{inverse automaton} of $\mathcal{T}$.
    
    \begin{example}\label{ex:addingMachine}
      The typical example of an automaton is generated by the \emph{adding machine}
      
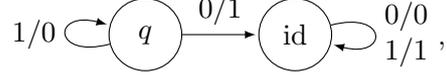
\begin{center}
        \begin{tikzpicture}[auto, shorten >=1pt, >=latex, baseline=(id.base)]
          \node[state] (q) {$q$};
          \node[state, right=of q] (id) {$\id$};
          \draw[->] (q) edge[loop left] node {$1/0$} (q)
                        edge node {$0/1$} (id)
                    (id) edge[loop right] node[align=left] {$0/0$\\$1/1$} (id);
        \end{tikzpicture},
      \end{center}
      which we denote by $\mathcal{T} = (\{ q, \id \}, \{ 0, 1 \}, \delta)$ in this example. It is deterministic, complete and invertible and the action of $\id$ is obviously the identity mapping on $\Sigma^*$. To understand the action of $q$, we observe that we have $q \circ 000 = 100$, $q \circ 100 = 010$ and $q \circ 010 = 110$. Thus, if we interpret a word $u \in \{ 0, 1 \}^*$ as the reverse/least-significant bit first binary representation of a natural number $n$, then $q \circ{}\!$ maps $u$ to the reverse/least-significant bit first binary representation of $n + 1$ (with appropriately many leading zeros). Therefore, the element $q \circ{}\!$ of the semigroup $\mathscr{S}(\mathcal{T})$ can be identified with plus one in the monoid of natural numbers with addition as operation; accordingly, $q^i \circ{}\!$ is plus $i$. Since we also have the identity as a state, the semigroup $\mathscr{S}(\mathcal{T})$ generated by $\mathcal{T}$ is isomorphic to $\mathbb{N}$ (with addition and including zero) or -- in different words -- the free monoid of rank one.
      
      Since the automaton is complete and invertible, we can also consider the group $\mathscr{G}(\mathcal{T})$ generated by it. The inverse of $q \circ{}\!$ can, obviously, be identified with minus one and we obtain that $\mathscr{G}(\mathcal{T})$ is isomorphic to the free group of rank one or the set of integers with addition as operation.
    \end{example}
    
    \paragraph{Dual Automaton and the Dual Action.}
    The \emph{dual} of an automaton $\mathcal{T} = (Q, \Sigma, \delta)$ is the automaton $\partial \mathcal{T} = (\Sigma, Q, \partial \delta)$ with
    \[
      \partial \delta = \{ \trans{a}{p}{q}{b} \mid \trans{p}{a}{b}{q} \in \delta \} \text{,}
    \]
    i.\,e.\ we swap the roles of $Q$ and $\Sigma$. To obtain a cross diagram for $\partial \mathcal{T}$ from one for $\mathcal{T}$, we have to mirror it at the north-west to south-east diagonal, i.\,e.\ we have the following equivalence of cross diagrams:
    \begin{center}
      \begin{tikzpicture}[baseline=(m-4-1.base)]
        \matrix[matrix of math nodes, text height=1.25ex, text depth=0.25ex] (m) {
                 & a_1 &    & \dots &    & a_m &     \\
          p_1    &     & {} & \dots & {} &     & q_1 \\
                 & {}  &    &       &    & {}  &     \\[-1.25ex]
          \vdots & \vdots &    &       &    & \vdots & \vdots \\[-1.5ex]
                 & {}  &    &       &    & {}  &     \\
          p_n    &     & {} & \dots & {} &     & q_n \\
                 & b_1 &    & \dots &    & b_m &     \\
        };
        \foreach \j in {1, 5} {
          \foreach \i in {1, 5} {
            \draw[->] let
              \n1 = {int(2+\i)},
              \n2 = {int(1+\j)}
            in
              (m-\n2-\i) -> (m-\n2-\n1);
            \draw[->] let
              \n1 = {int(1+\i)},
              \n2 = {int(2+\j)}
            in
              (m-\j-\n1) -> (m-\n2-\n1);
          };
        };
      \end{tikzpicture} in $\mathcal{T} \iff {}$
      \begin{tikzpicture}[baseline=(m-4-1.base)]
        \matrix[matrix of math nodes, text height=1.25ex, text depth=0.25ex] (m) {
                 & p_1 &    & \dots &    & p_n &     \\
          a_1    &     & {} & \dots & {} &     & b_1 \\
                 & {}  &    &       &    & {}  &     \\[-1.25ex]
          \vdots & \vdots &    &       &    & \vdots & \vdots \\[-1.5ex]
                 & {}  &    &       &    & {}  &     \\
          a_m    &     & {} & \dots & {} &     & b_m \\
                 & q_1 &    & \dots &    & q_n &     \\
        };
        \foreach \j in {1, 5} {
          \foreach \i in {1, 5} {
            \draw[->] let
              \n1 = {int(2+\i)},
              \n2 = {int(1+\j)}
            in
              (m-\n2-\i) -> (m-\n2-\n1);
            \draw[->] let
              \n1 = {int(1+\i)},
              \n2 = {int(2+\j)}
            in
              (m-\j-\n1) -> (m-\n2-\n1);
          };
        };
      \end{tikzpicture} in $\partial \mathcal{T}$.
    \end{center}
    If we let $\bm{p} = p_n \dots p_1$ and $\bm{q} = q_n \dots q_1$ as well as $u = a_1 \dots a_m$ and $v = b_1 \dots b_m$, we can write the above equivalence in short-hand notation:\enlargethispage{2\baselineskip}
    \begin{center}
      \begin{tikzpicture}[baseline=(m-2-3.base)]
        \matrix[matrix of math nodes, text height=1.25ex, text depth=0.25ex] (m) {
                 & u & \\
          \bm{p} &   & \bm{q} \\
                 & v & \\
        };
        \foreach \i in {1} {
          \draw[->] let
            \n1 = {int(2+\i)}
          in
            (m-2-\i) -> (m-2-\n1);
          \draw[->] let
            \n1 = {int(1+\i)}
          in
            (m-1-\n1) -> (m-3-\n1);
        };
      \end{tikzpicture} in $\mathcal{T} \iff {}$
      \begin{tikzpicture}[baseline=(m-2-3.base)]
        \matrix[matrix of math nodes, text height=1.25ex, text depth=0.25ex] (m) {
                  & \rev{\bm{p}} & \\
          \rev{u} &              & \rev{v} \\
                  & \rev{\bm{q}} & \\
        };
        \foreach \i in {1} {
          \draw[->] let
            \n1 = {int(2+\i)}
          in
            (m-2-\i) -> (m-2-\n1);
          \draw[->] let
            \n1 = {int(1+\i)}
          in
            (m-1-\n1) -> (m-3-\n1);
        };
      \end{tikzpicture} in $\partial \mathcal{T}$.
    \end{center}
    
    Clearly, taking the dual of an automaton is an involution and the dual of a deterministic (complete) automaton is also deterministic (complete). Additionally, the dual of an invertible automaton is reversible. Thus, the dual of an \SAut is an \SAut and the dual of a \GAut is a complete and reversible \SAut (and vice versa).
    
    Therefore, if $\mathcal{T} = (Q, \Sigma, \delta)$ is an \SAut, $\mathcal{T}$ itself induces the actions $\circ_{\mathcal{T}}$ and $\cdot_{\mathcal{T}}$ and its dual induces the actions $\circ_{\partial \mathcal{T}}$ and $\cdot_{\partial \mathcal{T}}$, which we simply write as $\circ_{\partial}$ and $\cdot_{\partial}$ if the automaton is clear from the context. Because of the above equivalence of cross diagrams, there is a strong connection between $\circ$ and $\cdot_{\partial}$ (and, equivalently, between $\circ_{\partial}$ and $\cdot$). We have
    \[
      \rev{u} \circ_{\partial} \rev{\bm{p}} = \rev{(\bm{p} \cdot u)} \text{ (or both undefined)}
    \]
    for all $u \in \Sigma^*$ and $\bm{p} \in Q^*$. Here, we have used the convention that $\partial$ has higher precedence than the two automaton actions to avoid parentheses; for example, $\partial \bm{p} \circ u$ is to be understood as $(\partial \bm{p}) \circ u$ instead of $\partial (\bm{p} \circ u)$.
    
    Since the dual of a complete and reversible \SAut $\mathcal{T} = (Q, \Sigma, \delta)$ is a \GAut and we, thus, have that all $u \circ_{\partial} {}\!$ are bijections, we immediately obtain the following fact from the above connection.
    \begin{fact}\label{fact:reversibleIsBijection}
      Let $\mathcal{T} = (Q, \Sigma, \delta)$ be a complete and reversible \SAut. Then, all maps $\!{} \cdot u: Q^* \cup Q^\omega \to Q^* \cup Q^\omega, \bm{p} \mapsto \bm{p} \cdot u$ for $u \in \Sigma^*$ are length-preserving bijections.
    \end{fact}
    
    \paragraph{Orbits.}
    Let $\mathcal{T} = (Q, \Sigma, \delta)$ be an \SAut and $u \in \Sigma^* \cup \Sigma^\omega$. For $K \subseteq Q^*$, the \emph{$K$-orbit} of $u$ is
    \[
      K \circ u = \{ \bm{q} \circ u \mid \bm{q} \in K, \bm{q} \circ u \text{ defined} \} \text{.}
    \]
    The \emph{orbit} of $u$ is its $Q^*$-orbit. The orbit $Q^* \circ u$ of a word $u$ has a natural graph structure: we have a $q$-labeled edge for $q \in Q$ from $\bm{p} \circ u$ to $q \bm{p} \circ u$ whenever $q \bm{p} \circ u$ is defined (where $\bm{p} \in Q^*$).
    
    If $\mathcal{T}$ is even a \GAut, we can define the $K$-orbit of $u$ for $K \subseteq Q^{\pm *}$ analogously. It is well-known that, for a \GAut $\mathcal{T} = (Q, \Sigma, \delta)$, the orbit of $u$ is infinite if and only if its $Q^{\pm *}$-orbit is infinite (see, e.\,g.\ \cite[Lemma~2.5]{decidabilityPart}).
    
    Recently, it could be shown that the existence of an $\omega$-word with an infinite orbit is equivalent to the algebraic property that an automaton semigroup (or group) is infinite.
    \begin{thm}[{\cite[Corollary~3.3]{orbitsPart}}]\label{cor:infiniteSemigroupsHaveInfiniteOrbits}
      The semigroup $\mathscr{S}(\mathcal{T})$ generated by some \SAut $\mathcal{T} = (Q, \Sigma, \delta)$ is infinite if and only if there is some $\omega$-word $\alpha \in \Sigma^\omega$ with an infinite orbit $Q^* \circ \alpha$.
    \end{thm}
    
    In fact, this result is only a special case of a more general result.
    \begin{thm}[{\cite[Theorem~3.2]{orbitsPart}}]\label{thm:infiniteSubsetHasInfiniteOrbit}
      Let $\mathcal{T} = (Q, \Sigma, \delta)$ be some \SAut and let $K \subseteq Q^*$ be suffix-closed. The image of $K$ in $\mathscr{S}(\mathcal{T})$ is infinite if and only if there is some $\omega$-word $\alpha$ whose $K$-orbit $K \circ \alpha$ is infinite.
    \end{thm}
    
    The proof of this result heavily relies on a dual argument. In fact, in the course of its proof, there appears a result which seems to be more fundamental than the actual end result and we will need this intermediate result directly for some of our proof below. In order to state it, we first need to introduce a concept that generalizes the concept of equality in an automaton semigroup. For an \SAut $\mathcal{T} = (Q, \Sigma, \delta)$ and a language $L \subseteq \Sigma^*$, we define the relation ${\equiv_{\mathcal{T}, L}} \subseteq Q^* \times Q^*$ by
    \[
      \bm{p} \equiv_{\mathcal{T}, L} \bm{q} \iff \forall u \in L: \bm{p} \circ u = \bm{q} \circ u \text{ (or both undefined)} \text{.}
    \]
    As is the case with the two automaton actions, we do not write the index $\mathcal{T}$ whenever the automaton is clear from the context. It is easy to verify that $\equiv_L$ is an equivalence relation for all languages $L \subseteq \Sigma^*$. We write $[\bm{p}]_L$ for the equivalence class of $\bm{p} \in Q^*$ under $\equiv_L$ and
    \[
      K / L = \{ [\bm{p}]_L \mid \bm{p} \in K \}
    \]
    for the set of classes with a representative in $K \subseteq Q^*$.
    The set $K / L$ is a generalization of a couple of other notions (including the automaton semigroup itself and co-sets with respect to stabilizers in the group case; see \cite{orbitsPart} for these examples) but, most importantly, it is closely related to the $K$-orbit of a word.
    \begin{prop}[{\cite[Proposition~2.4]{orbitsPart}}]\label{prop:KSchreierAndKOrbit}
      Let $\mathcal{T} = (Q, \Sigma, \delta)$ be an \SAut and let $K \subseteq Q^*$ be suffix-closed.
      For all $\alpha \in \Sigma^\omega$, we have
      \[
        | K / \Pre \alpha | = \infty \iff | K \circ \alpha | = \infty \text{.}
      \]
    \end{prop}\noindent
    In particular, this result also holds for non-complete \SAuta.
    
    For an \SAut $\mathcal{T} = (Q, \Sigma, \delta)$ and a set $K \subseteq Q^*$, we also have the equivalence $\equiv_{K, \partial \mathcal{T}}$ belonging to the dual of $\mathcal{T}$. We also simply write $\equiv_K$ for this relation and $L / K$ for the classes of $\equiv_K$ with a representative in $L \subseteq \Sigma^*$. That these are to be understood with respect to $\partial \mathcal{T}$ (and not with respect to $\mathcal{T}$ itself) can be seen from the fact that $K$ is a set of state sequences of $\mathcal{T}$.
    
    The above-mentioned result shows a close relation between the sets $K / L$ and $L / K$:
    \begin{prop}[{\cite[Proposition~3.1]{orbitsPart}}]\label{prop:KLInfiniteIffLKInfinite}
      Let $\mathcal{T} = (Q, \Sigma, \delta)$ be an \SAut, let $K \subseteq Q^*$ be suffix-closed and let $L \subseteq \Sigma^*$ be prefix-closed. Then, we have:
      \[
      | K / L | = \infty \iff | \partial L / \partial K | = \infty
      \]
    \end{prop}
    
    An important special case of \autoref{prop:KLInfiniteIffLKInfinite} is when $K$ and $L$ are both given by a single $\omega$-word. Combined with \autoref{prop:KSchreierAndKOrbit}, this case yields the following duality result for orbits (which we will also use below).
    \begin{cor}[{\cite[Corollary~3.11]{orbitsPart}}]\label{cor:alphaAndPiOrbit}
      Let $\mathcal{T} = (Q, \Sigma, \delta)$ be an \SAut and let $\pi \in Q^\omega$ and $\alpha \in \Sigma^\omega$. Then, we have
      \[
        |\partial \Pre \pi \circ \alpha| = \infty \iff |\partial \Pre \alpha \circ_\partial \pi| = \infty
      \]
    \end{cor}
    
    \autoref{prop:KLInfiniteIffLKInfinite} can also be used to prove the following connection between an element of an automaton semigroup having torsion and the infinity of the corresponding dual orbit (compare to \cite[Theorem~3]{DaRo16} and \cite[Proposition 7]{KPS}).
    \begin{thm}[{\cite[Theorem~3.12]{orbitsPart}}]\label{thm:torsionIsFiniteDualOrbit}
      Let $\mathcal{T} = (Q, \Sigma, \delta)$ be an \SAut and let $\bm{q} \in Q^+$. Then, the statements
      \begin{enumerate}
        \item $\partial \bm{q}$ has torsion in $\mathscr{S}(\mathcal{T})$.
        \item The orbit $\Sigma^* \circ_{\partial} \bm{q}^\omega$ of $\bm{q}^\omega$ under the action of the dual of $\mathcal{T}$ is finite.
        \item The orbit $\Sigma^* \circ_{\partial} \bm{p} \bm{q}^\omega$ of $\bm{p} \bm{q}^\omega$ under the action of the dual of $\mathcal{T}$ is finite for all $\bm{p} \in Q^*$.
      \end{enumerate}
      are equivalent.
    \end{thm}
  \end{section}

  \begin{section}{The Finiteness Problem for Automaton Groups}
    The question whether the finiteness problem for automaton groups
    \problem{
      a \GAut $\mathcal{T}$
    }{
      is $\mathscr{G}(\mathcal{T})$ finite?
    }\noindent
    is undecidable is an important open question in the algorithmic study of automaton groups \cite[7.2 b)]{GriNeShu}. In this section, we will show that a more general version of the problem is undecidable. We will show this by using Gillibert's result that there is an automaton group whose order problem
    \problem[
      a \GAut $\mathcal{T} = (Q, \Sigma, \delta)$
    ]{
      a finite state sequence $\bm{q} \in Q^*$
    }{
      has $\bm{q}$ finite order in $\mathscr{G}(\mathcal{T})$?
    }\noindent
    is undecidable \cite{gillibert2018automaton}. In fact, this result was also obtain by Bartholdi and Mitrofanov \cite{bartholdi2017wordAndOrderProblems} but we specifically use the construction given by Gillibert.
    
    \begin{thm}\label{thm:generalizedGroupFinitenessProblem}
      The decision problem
      \problem[
        a \GAut $\mathcal{T} = (Q, \Sigma, \delta)$
      ]{
        a finite word $w \in \Sigma^*$
      }{
        is $\mathscr{G}(\mathcal{T}) \cdot w = \{ g \cdot w \mid g \in \mathscr{G}(\mathcal{T}) \}$ finite?
      }\noindent
      is undecidable for some \GAut $\mathcal{T}$.
    \end{thm}
    \begin{proof}
      Although it is not explicitly stated in his proof, Gillibert actually shows the undecidability of the decision problem\enlargethispage{\baselineskip}
      \problem[
        a \GAut $\mathcal{R} = (P, \Gamma, \tau)$ and\newline
        a state $\$ \in P$
      ]{
        a finite sequence $\bm{p} \in P^*$ of states
      }{
        has $\$\Lambda(\bm{p})$ finite order in $\mathscr{G}(\mathcal{R})$?
      }\noindent
      where $\Lambda: P^* \to P^*$ is given by $\Lambda(\varepsilon) = \varepsilon$ and $\Lambda(\hat{p} \bm{p}) = \Lambda(\bm{p}) \hat{p} \Lambda(\bm{p})$ \cite{gillibert2018automaton}.\footnote{In Gillibert's paper, the function is called $Q$, actually. However, this notation clashes with the convention of using $Q$ to denote the state sets of automata followed in this work. Therefore, we use $\Lambda$ instead. Additionally, Gillibert uses right actions to define automaton groups. Therefore, we mirror the ordering here.}

      We take the \GAut $\mathcal{R}$ and extend it into a new \GAut $\mathcal{T} = (Q, \Sigma, \delta)$. Then, we reduce the above version of the order problem of $\mathcal{R}$ to the generalized finiteness problem for $\mathcal{T}$ from the theorem statement.
      
      As the alphabet of $\mathcal{T}$, we use $\Sigma = \Gamma \uplus \{ a_p \mid p \in P \} \times \{ 0, 1 \} \uplus \{ *, \# \}$, i.\,e.\ we add two new special letters $*$ and $\#$ as well as two new letters $(a_p, 0)$ and $(a_p, 1)$ for every state $p \in P$. Similarly, we use $Q = P \uplus \{ s, t, \id \} \uplus \{ \#_p \mid p \in P \}$ for the state set, i.\,e.\ we add three new states $s$, $t$ and $\id$ as well as a new state $\#_p$ for every old state $p \in P$. Of course, we also add new transitions
      \begin{align*}
        \delta' &= \tau \cup \{ \trans{s}{*}{*}{t}, \trans{t}{\#}{\#}{\$} \} \cup \{ \trans{t}{(a_p, 1)}{(a_p, 0)}{t}, \trans{t}{(a_p, 0)}{(a_p, 1)}{\#_p} \mid p \in P \} \\
        &\cup \{ \trans{\#_p}{(a_q, i)}{(a_q, i)}{\#_p}, \trans{\#_p}{\#}{\#}{p} \mid p, q \in P, i \in \{ 0, 1 \} \}\\
        &\cup \{ \trans{\id}{a}{a}{\id} \mid a \in \Sigma \} \text{,}
      \end{align*}
      which are depicted schematically in \autoref{fig:transitionsForInverseSemigroupOrbitFiniteness}, and make the automaton complete by adding a transition to the identity state whenever some transition is missing:
      \begin{align*}
        \delta &= \delta' \cup \{ \trans{q}{a}{a}{\id} \mid q \in Q, a \in \Sigma, \nexists a' \in \Sigma, q' \in Q: \trans{q}{a}{a'}{q'} \in \delta' \}
      \end{align*}%
      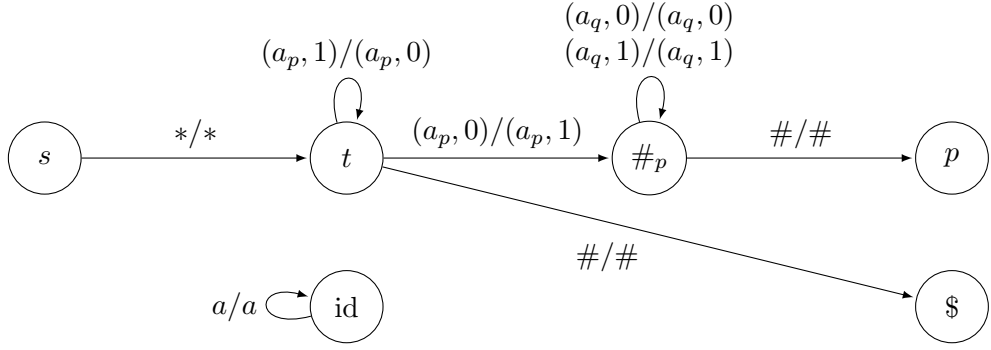
\begin{figure}\centering
        \begin{tikzpicture}[auto, shorten >=1pt, >=latex, node distance=1cm and 3cm]
          \node[state] (s) {$s$};
          \node[state, right=of s] (t) {$t$};
          \node[state, right=of t] (sharp) {$\#_p$};
          \node[state, right=of sharp] (p) {$p$};
          \node[state, below=of p] (dollar) {$\$$};
          \node[state, below=of t] (id) {$\id$};

          \draw[->] (s) edge node {$*/*$} (t)
                    (t) edge[loop above] node {$(a_p, 1) / (a_p, 0)$} (t)
                        edge node {$(a_p, 0) / (a_p, 1)$} (sharp)
                        edge node[swap] {$\# / \#$} (dollar)
                    (sharp) edge[loop above] node[align=center] {$(a_q, 0) / (a_q, 0)$\\$(a_q, 1) / (a_q, 1)$} (sharp)
                            edge node {$\# / \#$} (p)
                    (id) edge[loop left] node {$a / a$} (id)
          ;
        \end{tikzpicture}
        \caption{New transitions of $\mathcal{T}$}\label{fig:transitionsForInverseSemigroupOrbitFiniteness}
      \end{figure}%
      Note that the resulting automaton is indeed a \GAut!

      For the reduction of the strengthened version of the order problem to the generalized version of the finiteness problem, we map the input sequence $\bm{p} = p_\ell \dots p_1$ to the finite word $w = * w' = * (a_{p_1}, 0) \dots (a_{p_\ell}, 0) \#$, which is obviously computable. In the remainder of this proof, we show that $\$ \Lambda(\bm{p})$ has finite order in $\mathscr{G}(\mathcal{R})$ if and only if $\mathscr{G}(\mathcal{T}) \cdot w$ is finite.
      
      First, we show that $\$ \Lambda(\bm{p})$ has finite order in $\mathscr{G}(\mathcal{R})$ if and only if it has in $\mathscr{G}(\mathcal{T})$. We do this, by showing that $(\$ \Lambda(\bm{p}))^i$ and $(\$ \Lambda(\bm{p}))^j$ are distinct in $\mathscr{G}(\mathcal{R})$ if and only if they are distinct in $\mathscr{G}(\mathcal{T})$ for $i, j \in \mathbb{N}$. If they are distinct in $\mathscr{G}(\mathcal{R})$, there is some witness $u \in \Gamma^*$ which they act differently on. Since we have $\tau \subseteq \delta' \subseteq \delta$, this is also a witness for their difference in $\mathscr{G}(\mathcal{T})$. For the other direction, suppose that $(\$ \Lambda(\bm{p}))^i$ is different to $(\$ \Lambda(\bm{p}))^j$ in $\mathscr{G}(\mathcal{T})$. Then, there must be some witness $u \in \Sigma^*$ which they act differently on. We are done if $u$ is already in $\Gamma^*$. Otherwise, we can factorize $u = u_1 a u_2$ with $u_1 \in \Gamma^*$, $a \in \Sigma \setminus \Gamma$ and $u_2 \in \Sigma^*$. By the construction of $\mathcal{T}$, we remain in states from $P$ if we start in $P$ and read letters from $\Gamma$. If we read a letter from $\Sigma \setminus \Gamma$, we go to $\id$, which yields the cross diagrams
      \begin{center}
        \begin{tikzpicture}[baseline=(m-2-1.base)]
          \matrix[matrix of math nodes, text height=1.25ex, text depth=0.25ex] (m) {
                                   & u_1 &    & a &     & u_2 &     \\
            (\$ \Lambda(\bm{p}))^i &     & {} &   & \id^{i \, |\$ \Lambda(\bm{p})|} &   & {} \\
                                   & v_1 &    & a &     & u_2 &     \\
          };
          \foreach \j in {1} {
            \foreach \i in {1, 3, 5} {
              \draw[->] let
                \n1 = {int(2+\i)},
                \n2 = {int(1+\j)}
              in
                (m-\n2-\i) -> (m-\n2-\n1);
              \draw[->] let
                \n1 = {int(1+\i)},
                \n2 = {int(2+\j)}
              in
                (m-\j-\n1) -> (m-\n2-\n1);
            };
          };
        \end{tikzpicture}
        and
        \begin{tikzpicture}[baseline=(m-2-1.base)]
          \matrix[matrix of math nodes, text height=1.25ex, text depth=0.25ex] (m) {
                                   & u_1  &    & a &     & u_2 &     \\
            (\$ \Lambda(\bm{p}))^j &      & {} &   & \id^{j \, |\$ \Lambda(\bm{p})|} &   & {} \\
                                   & v_1' &    & a &     & u_2 &     \\
          };
          \foreach \j in {1} {
            \foreach \i in {1, 3, 5} {
              \draw[->] let
                \n1 = {int(2+\i)},
                \n2 = {int(1+\j)}
              in
                (m-\n2-\i) -> (m-\n2-\n1);
              \draw[->] let
                \n1 = {int(1+\i)},
                \n2 = {int(2+\j)}
              in
                (m-\j-\n1) -> (m-\n2-\n1);
            };
          };
        \end{tikzpicture}
        for $\mathcal{T}$.
      \end{center}
      Thus, $(\$ \Lambda(\bm{p}))^i$ and $(\$ \Lambda(\bm{p}))^j$ must already act differently on $u_1$, which is from $\Gamma^*$ and, thus, also a witness for $\mathcal{R}$.
      
      Next, we observe that $*$ is not changed by the action of any state and that we have $q \cdot * = \id$ for all $q \in Q$ except $s$ and $s \cdot * = t$. Thus, $\mathscr{G}(\mathcal{T}) \cdot *$ is the subgroup $T$ generated by $t$ in $\mathscr{G}(\mathcal{T})$ and we obtain $\mathscr{G}(\mathcal{T}) \cdot w = T \cdot w'$. To understand the elements in $T \cdot w'$, we will show that we have the cross diagram
      \begin{center}
        \hspace*{\fill}\begin{tikzpicture}[baseline=(m-2-1.base)]
          \matrix[matrix of math nodes, text height=1.25ex, text depth=0.25ex] (m) {
                               & w' & \\
            t^{k \, |\$ \Lambda(\bm{p})|} &   & (\$ \Lambda(\bm{p}))^k \\
                               & w' & \\
          };
          \foreach \j in {1} {
            \foreach \i in {1} {
              \draw[->] let
                \n1 = {int(2+\i)},
                \n2 = {int(1+\j)}
              in
                (m-\n2-\i) -> (m-\n2-\n1);
              \draw[->] let
                \n1 = {int(1+\i)},
                \n2 = {int(2+\j)}
              in
                (m-\j-\n1) -> (m-\n2-\n1);
            };
          };
        \end{tikzpicture}
        in $\mathcal{T}$\hfill\makebox[0pt]{\hypertarget{finiteness:claim}{($\dagger$)}}
      \end{center}
      for all $k \in \mathbb{N}$. This shows that $\mathscr{G}(\mathcal{T}) \cdot w = T \cdot w'$ is given by the state sequences
      \[
        \left( \Suf \$ \Lambda(\bm{p}) \right) \left( \$ \Lambda(\bm{p}) \right)^*
      \]
      and their inverses in $\mathscr{G}(\mathcal{T})$. These form a finite set in $\mathscr{G}(\mathcal{T})$ if and only if $\left( \$ \Lambda(\bm{p}) \right)^*$ is finite in $\mathscr{G}(\mathcal{T})$ and this is the case if and only if $\$ \Lambda(\bm{p})$ has finite order in $\mathscr{G}(\mathcal{T})$ (or, equivalently, in $\mathscr{G}(\mathcal{R})$).

      The easiest way to establish the cross diagrams \hyperlink{finiteness:claim}{($\dagger$)} is by calculation. For example, for $\bm{p} = p_3 p_2 p_1$, we have $w' = (a_{p_1}, 0)(a_{p_2}, 0)(a_{p_3}, 0)\#$ and the cross diagram:
      \begin{center}
        \begin{tikzpicture}
          \matrix[matrix of math nodes, text height=1.25ex, text depth=0.25ex] (m) {
              & (a_{p_1}, 0) &          & (a_{p_2}, 0) &          & (a_{p_3}, 0) &          & \# & \\
            t &              & \#_{p_1} &              & \#_{p_1} &              & \#_{p_1} &    & p_1 \\
              & (a_{p_1}, 1) &          & (a_{p_2}, 0) &          & (a_{p_3}, 0) &          & \# & \\
            t &              & t        &              & \#_{p_2} &              & \#_{p_2} &    & p_2 \\
              & (a_{p_1}, 0) &          & (a_{p_2}, 1) &          & (a_{p_3}, 0) &          & \# & \\
            t &              & \#_{p_1} &              & \#_{p_1} &              & \#_{p_1} &    & p_1 \\
              & (a_{p_1}, 1) &          & (a_{p_2}, 1) &          & (a_{p_3}, 0) &          & \# & \\
            t &              & t        &              & t        &              & \#_{p_3} &    & p_3 \\
              & (a_{p_1}, 0) &          & (a_{p_2}, 0) &          & (a_{p_3}, 1) &          & \# & \\
            t &              & \#_{p_1} &              & \#_{p_1} &              & \#_{p_1} &    & p_1 \\
              & (a_{p_1}, 1) &          & (a_{p_2}, 0) &          & (a_{p_3}, 1) &          & \# & \\
            t &              & t        &              & \#_{p_2} &              & \#_{p_2} &    & p_2 \\
              & (a_{p_1}, 0) &          & (a_{p_2}, 1) &          & (a_{p_3}, 1) &          & \# & \\
            t &              & \#_{p_1} &              & \#_{p_1} &              & \#_{p_1} &    & p_1 \\
              & (a_{p_1}, 1) &          & (a_{p_2}, 1) &          & (a_{p_3}, 1) &          & \# & \\
            t &              & t        &              & t        &              & t        &    & \$ \\
              & (a_{p_1}, 0) &          & (a_{p_2}, 0) &          & (a_{p_3}, 0) &          & \# & \\
          };
          \foreach \j in {1, 3, 5, 7, 9, 11, 13, 15} {
            \foreach \i in {1, 3, 5, 7} {
              \draw[->] let
                \n1 = {int(2+\i)},
                \n2 = {int(1+\j)}
              in
                (m-\n2-\i) -> (m-\n2-\n1);
              \draw[->] let
                \n1 = {int(1+\i)},
                \n2 = {int(2+\j)}
              in
                (m-\j-\n1) -> (m-\n2-\n1);
            };
          };
          \draw[decorate, decoration=brace] (m-2-9.north east) -- node[right] (label) {$\Lambda(q_2 q_1)$} (m-6-9.south east);
          \draw[decorate, decoration=brace] (m-10-9.north east) -- node[right] {$\Lambda(q_2 q_1)$} (m-14-9.south east);
          \draw[decorate, decoration=brace] (m-2-9.north east -| label.east) -- node[right] {$\Lambda(q_3 q_2 q_1)$} (m-14-9.south east -| label.east);
        \end{tikzpicture}
      \end{center}
      Notice that, in the second component, $t$ implements a binary increment (in the same way as the adding machine in \autoref{ex:addingMachine}). This is what creates the pattern of $\Lambda(\bm{p})$.

      For a formal proof, we first define the shorthand notations $(a_{\bm{p}}, i) = (a_{p_1}, i) \dots (a_{p_\ell}, i)$ for $i \in \{ 0, 1 \}$ and $\bm{p} = p_\ell \dots p_1$ as well as $\#_\varepsilon = \varepsilon$ and $\#_{\Lambda(\hat{p} \bm{p})} = \#_{\Lambda(\bm{p})} \#_{\hat{p}} \#_{\Lambda(\bm{p})}$ for $\bm{p} \in P^*$ and $\hat{p} \in P$. We start by showing the cross diagram(s)
      \begin{center}
        \begin{tikzpicture}
          \matrix[matrix of math nodes, text height=1.25ex, text depth=0.25ex] (m) {
                            & (a_{\bm{p}}, 0) & \\
            t^{|\Lambda(\bm{p})|} &                 & \#_{\Lambda(\bm{p})} \\
                            & (a_{\bm{p}}, 1) & \\
                          t &                 & t \\
                            & (a_{\bm{p}}, 0) & \\
          };
          \foreach \j in {1, 3} {
            \foreach \i in {1} {
              \draw[->] let
                \n1 = {int(2+\i)},
                \n2 = {int(1+\j)}
              in
                (m-\n2-\i) -> (m-\n2-\n1);
              \draw[->] let
                \n1 = {int(1+\i)},
                \n2 = {int(2+\j)}
              in
                (m-\j-\n1) -> (m-\n2-\n1);
            };
          };
        \end{tikzpicture}
      \end{center}
      for every $\bm{p} \in P^+$ by induction on the length of $\bm{p}$. For $\bm{p} = p \in P$, this is easily verified from the definition of $\mathcal{T}$ (note that $\Lambda(\bm{p}) = \Lambda(p) = p$ in this case). For $\bm{p}' = \hat{p} \bm{p}$ with $\hat{p} \in P$, we have $|\Lambda(\hat{p} \bm{p})| = 2|\Lambda(\bm{p})| + 1$ and the cross diagram
      \begin{center}
        \begin{tikzpicture}
          \matrix[matrix of math nodes, text height=1.25ex, text depth=0.25ex] (m) {
                            & (a_{\bm{p}}, 0) &                & (a_{\hat{p}}, 0) & \\
            t^{|\Lambda(\bm{p})|} &                 & \#_{\Lambda(\bm{p})} &          & \#_{\Lambda(\bm{p})} \\
                            & (a_{\bm{p}}, 1) &                & (a_{\hat{p}}, 0) & \\
                          t &                 & t              &          & \#_{\hat{p}}\\
                            & (a_{\bm{p}}, 0) &                & (a_{\hat{p}}, 1) & \\
            t^{|\Lambda(\bm{p})|} &                 & \#_{\Lambda(\bm{p})} &          & \#_{\Lambda(\bm{p})} \\
                            & (a_{\bm{p}}, 1) &                & (a_{\hat{p}}, 1) & \\
                          t &                 & t              &          & t \\
                            & (a_{\bm{p}}, 0) &                & (a_{\hat{p}}, 0) & \\
          };
          \foreach \j in {1, 3, 5, 7} {
            \foreach \i in {1, 3} {
              \draw[->] let
                \n1 = {int(2+\i)},
                \n2 = {int(1+\j)}
              in
                (m-\n2-\i) -> (m-\n2-\n1);
              \draw[->] let
                \n1 = {int(1+\i)},
                \n2 = {int(2+\j)}
              in
                (m-\j-\n1) -> (m-\n2-\n1);
            };
          };

          \path[fill=gray, opacity=0.2] (m-1-2.north -| m-2-1.west) rectangle (m-9-2.south -| m-2-3.east);
          \draw[dotted] (m-5-2.west -| m-2-1.west) -- (m-5-2.west)
                        (m-5-2.east) -- (m-5-2.west -| m-2-3.east);
          \draw[decorate, decoration=brace] (m-2-5.north east) -- node[right] (label) {$\#_{\Lambda(\hat{p} \bm{p})}$} (m-6-5.south east);
        \end{tikzpicture}
      \end{center}
      where the shaded part is obtained from using the induction hypothesis twice. For the part on the right, notice that we have $\#_p \circ (a_q, i) = (a_q, i)$ and $\#_p \cdot (a_q, i) = \#_p$ for all $p, q \in P$ and $i \in \{ 0, 1 \}$ by construction. The two transactions on the right involving $t$ can directly be verified, which concludes the induction.

      Finally, we can extend this to prove the cross diagrams \hyperlink{finiteness:claim}{($\dagger$)} required above:
      \begin{center}
        \begin{tikzpicture}
          \matrix[matrix of math nodes, text height=1.25ex, text depth=0.25ex] (m) {
              & (a_{\bm{p}}, 0) & & \# & \\
            t^{|\Lambda(\bm{p})|} & & \#_{\Lambda(\bm{p})} & & {\Lambda(\bm{p})} \\
              & (a_{\bm{p}}, 1) & & \# & \\
            t & & t & & \$ \\
              & (a_{\bm{p}}, 0) & & \# & \\
          };
          \foreach \j in {1, 3} {
            \foreach \i in {1, 3} {
              \draw[->] let
                \n1 = {int(2+\i)},
                \n2 = {int(1+\j)}
              in
                (m-\n2-\i) -> (m-\n2-\n1);
              \draw[->] let
                \n1 = {int(1+\i)},
                \n2 = {int(2+\j)}
              in
                (m-\j-\n1) -> (m-\n2-\n1);
            };
          };
        \end{tikzpicture}
      \end{center}
      The only point to notice here is that we indeed have $\#_{\Lambda(\bm{p})} \cdot \# = \Lambda(\bm{p})$; however, this is straight-forward to verify.
    \end{proof}
    
    \paragraph{Orbital and Dual Formulation.}
    The finiteness problem for automaton groups and the generalized problem from \autoref{thm:generalizedGroupFinitenessProblem} can also be formulated in other ways.
    
    The first one is a re-formulation based on orbits. Using \autoref{cor:infiniteSemigroupsHaveInfiniteOrbits}, we immediately obtain that the finiteness problem for automaton groups is equivalent to (the complement of) the problem:
    \problem{
      a \GAut $\mathcal{T} = (Q, \Sigma, \delta)$
    }{
      $\exists \alpha \in \Sigma^\omega: |Q^* \circ \alpha| = \infty$?
    }\noindent
    For the problem in \autoref{thm:generalizedGroupFinitenessProblem}, this view yields the following formulation.
    \begin{cor}\label{cor:groupWithUndecidablePotentialProblem}
      The decision problem
      \problem[
        a \GAut $\mathcal{T} = (Q, \Sigma, \delta)$
      ]{
        a finite word $w \in \Sigma^*$
      }{
        $\exists \alpha \in \Sigma^\omega: |Q^* \circ w\alpha| = \infty$?
      }\noindent
      is undecidable for some \GAut $\mathcal{T}$.
    \end{cor}
    \begin{proof}
      We have to show that $\mathscr{G}(\mathcal{T}) \cdot w$ is infinite if and only if the is some $\omega$-word $\alpha \in \Sigma^\omega$ such that the orbit $Q^* \circ w\alpha$ is infinite.
      
      In $\mathscr{G}(\mathcal{T})$, the elements of $\mathscr{G}(\mathcal{T}) \cdot w$ are given by $Q^{\pm *} \cdot w$, which is suffix-closed and, thus, by \autoref{thm:infiniteSubsetHasInfiniteOrbit}, infinite in $\mathscr{G}(\mathcal{T})$ if and only if there is some $\alpha \in \Sigma^\omega$ with $|Q^{\pm *} \cdot w \circ \alpha| = \infty$. We claim that $Q^{\pm *} \cdot w \circ \alpha$ is infinite if and only if $Q^{\pm *} \circ w \alpha$ is. Since the latter is the case if and only if $Q^* \circ w \alpha$ is infinite, we are done when we have shown this claim.
      
      Clearly, we can map $Q^{\pm *} \circ w \alpha$ surjectively onto $Q^{\pm *} \cdot w \circ \alpha$ by removing the prefix of length $|w|$. Thus, if $Q^{\pm *} \cdot w \circ \alpha$ is infinite, so must be $Q^{\pm *} \circ w \alpha$. On the other hand, we have $Q^{\pm *}* \circ w \alpha \subseteq \left( Q^{\pm *} \circ w \right) \left( Q^{\pm *} \cdot w \circ \alpha \right)$ and the first of the two sets on the right is always finite. Thus, if $Q^{\pm *} \circ w \alpha$ is infinite, $Q^{\pm *} \cdot w \circ \alpha$ must also be infinite.
    \end{proof}
    
    Another re-formulation is based on the dual automaton. A \GAut generates an infinite group if and only if its dual generates an infinite semigroup (see, e.\,g., \cite[Proof of Lemma~5]{aklmp12} combined with \cite[Proposition~10]{aklmp12}). Thus, the finiteness problem for automaton groups is equivalent to the problem
    \problem{
      a complete and reversible \SAut $\mathcal{T} = (Q, \Sigma, \delta)$
    }{
      is $\mathscr{S}(\mathcal{T})$ finite?
    }\noindent
    If we re-formulate the problem from \autoref{thm:generalizedGroupFinitenessProblem} under this view, we basically obtain the finiteness problem for left principal ideals for semigroups generated by complete and reversible \SAuta.
    \begin{cor}
      The decision problem
      \problem[
        a complete and reversible \SAut $\mathcal{R} = (P, \Gamma, \tau)$
      ]{
        a finite state sequence $\bm{p} \in P^*$
      }{
        is $P^* \bm{p}$ finite in $\mathscr{S}(\mathcal{R})$?
      }\noindent
      is undecidable.
    \end{cor}
    \begin{proof}
      We reduce the problem from \autoref{cor:groupWithUndecidablePotentialProblem} to this problem. As the automaton $\mathcal{R}$, we choose the dual $\partial \mathcal{T}$ of the \GAut $\mathcal{T} = (Q, \Sigma, \delta)$ and, for the reduction, we map $w \in \Sigma^*$ to $\partial w$ as the input sequence $\bm{p}$. We have to show that there is some $\alpha \in \Sigma^\omega$ with $|Q^* \circ w \alpha| = \infty$ if and only if $\Sigma^* \partial w$ is infinite in $\mathscr{S}(\partial \mathcal{T})$.

      If the orbital graph $Q^* \circ w \alpha$ is infinite (for some $\alpha \in \Sigma^\omega$), it must contain an infinite simple path starting in $w \alpha$. In other words, there is some $\pi \in Q^\omega$ with $\pi = p_1 p_2 \dots$ (where $p_1, p_2, \dots \in Q$) such that all $p_i \dots p_1 \circ w \alpha$ are pairwise distinct and that, thus, $\partial \Pre \pi \circ w \alpha$ is infinite.
      From \autoref{cor:alphaAndPiOrbit}, we obtain that $\partial \Pre (w \alpha) \circ_{\partial} \pi = \Suf \left( (\partial w) (\partial \alpha) \right) \circ_{\partial} \pi \subseteq \Sigma^* \partial w \circ_{\partial} \pi$ is also infinite, which shows that $\Sigma^* \partial w$ is infinite in $\mathscr{S}(\partial \mathcal{T})$.
      
      On the other hand, if $\Sigma^* \partial w$ is infinite in $\mathscr{S}(\partial \mathcal{T})$, we have in particular that $L = \Sigma^* \partial w \cup \Suf \partial w$ is infinite in $\mathscr{S}(\partial \mathcal{T})$. Since $L$ is suffix-closed, we obtain by \autoref{thm:infiniteSubsetHasInfiniteOrbit} that there is some $\pi \in Q^{\omega}$ such that $L \circ_{\partial} \pi$ is infinite. This is only possible if $\Sigma^* \partial w \circ_{\partial} \pi$ is infinite. Therefore,  we have an infinite path in the orbital graph $\Sigma^* \partial w \circ_{\partial} \pi$ which starts in $\pi$, goes to $\partial w \circ_{\partial} \pi$ and then continues as an infinite simple path. In other words, there has to be some $\alpha \in \Sigma^\omega$ with $\alpha = a_1 a_2 \dots$ (where $a_1, a_2, \dots \in \Sigma$) such that all $a_i \dots a_i (\partial w) \circ_{\partial} \pi$ are pairwise distinct. In particular, $\partial \Pre (w \alpha) \circ_{\partial} \pi$ is infinite, which implies that $\partial \Pre \pi \circ w \alpha \subseteq Q^* \circ w \alpha$ is also infinite by \autoref{cor:alphaAndPiOrbit}.
    \end{proof}
  \end{section}
  
  \begin{section}{Finite Orbits}\label{sec:FiniteOrbits}
    So far, we have looked at infinite orbits. In this section, we look at the opposite end and study $\omega$-words with finite orbits.
    
    While the existence of an infinite orbit is coupled to the algebraic property of an automaton semigroup $\mathscr{S}(\mathcal{T})$ to be infinite (by \autoref{cor:infiniteSemigroupsHaveInfiniteOrbits}), the existence of an $\omega$-word with a finite orbit depends on the generating automaton $\mathcal{T}$ (i.\,e.\ it is a property of the way the semigroup is presented, not an algebraic property). Indeed, if an \SAut $\mathcal{T}$ does not admit an $\omega$-word whose orbit is finite, we can add a new letter $a$ to the alphabet of $\mathcal{T}$ and loops $\trans{q}{a}{a}{q}$ to every state $q$. Obviously, this does not change the generated semigroup; however, now $a^\omega$ has a finite orbit.

    \paragraph{Periodic and Ultimately Periodic Words.}
    We have just seen that we can add transitions to any \SAut to even obtain a periodic and, thus, ultimately periodic $\omega$-word with a finite orbit (without changing the generated semigroup). We will see next that, in many cases, we do not even need to change the automaton: if there is an $\omega$-word with a finite orbit, then there is already an ultimately periodic word with finite orbit for every \SAut.\footnote{Contrasting this, we will later see in \autoref{prop:infiniteSemigroupWithAllPeriodicOrbitsFinite} that there are semigroups where the only infinite orbits belong to words that are not ultimately periodic.} If the \SAut is complete and reversible, we even have a periodic word with a finite orbit.

    \begin{prop}\label{prop:finiteOrbitImpliesPeriodicFiniteOrbit}
      Let $\mathcal{T} = (Q, \Sigma, \delta)$ be an \SAut. If there is an $\omega$-word $\alpha \in \Sigma^\omega$ such that its orbit $Q^* \circ \alpha$ is finite, then there are $u \in \Sigma^*$ and $v \in \Sigma^+$ such that $Q^* \circ u v^\omega$ is finite and $v$ can be chosen in such a way that it contains all letters that appear infinitely often in $\alpha$.
      
      If, in addition, $\mathcal{T}$ is complete and reversible, then we already have that the orbit $Q^* \circ v^\omega$ is finite.
    \end{prop}
    \begin{proof}
      Suppose we have $|Q^* \circ \alpha| < \infty$ for the $\omega$-word $\alpha = a_1 a_2 \dots$ with $a_1, a_2, \dots \in \Sigma$ and the \SAut $\mathcal{T}$. By \autoref{prop:KSchreierAndKOrbit}, this is equivalent to $|Q^* / \Pre \alpha| < \infty$, which, by \autoref{prop:KLInfiniteIffLKInfinite}, is equivalent to $|\partial \Pre \alpha / Q^*| < \infty$. Thus, there is an infinite set $I \subseteq \mathbb{N}$ with $a_i \dots a_1 \equiv_{Q^*} a_j \dots a_1$ for all $i, j \in I$. Let $k = \min I$ and $\ell = \min I \setminus \{ k \}$ and define
      \[
        u = a_1 \dots a_k \text{ and } v = a_{k + 1} \dots a_\ell \text{.}
      \]
      For this choice, we have that $\partial (u v^*) / Q^*$ contains only one element and that $\partial \Pre u v^\omega / Q^*\allowbreak = \{ [\partial w]_{Q^*} \mid w \in \Pre uv \}$ is still finite. By \autoref{prop:KLInfiniteIffLKInfinite}, this implies that $Q^* / \Pre uv^\omega$ is finite, which is the case if and only if $Q^* \circ uv^\omega$ is finite by \autoref{prop:KSchreierAndKOrbit}.
      
      If $\mathcal{T}$ is additionally complete and reversible, there is a surjective function $Q^* \circ uv^\omega \to Q^* \circ v^\omega$, which shows $|Q^* \circ v^\omega| \leq |Q^* \circ uv^\omega| < \infty$. This function maps a word $w \beta \in Q^* \circ uv^\omega$ with prefix $w$ of length $|w| = |u|$ to $\beta$. Clearly, $w \beta \in Q^* \circ uv^\omega$ implies $\beta \in Q^* \circ v^\omega$ and, to show that the function is surjective, consider an arbitrary element $\beta \in Q^* \circ v^\omega$. Then, there is some $\bm{q} \in Q^*$ with $\bm{q} \circ v^\omega = \beta$. Since $\mathcal{T}$ is complete and reversible, there is some $\bm{p} \in Q^{|\bm{q}|}$ with $\bm{p} \cdot u = \bm{q}$ (as the map $\!{}\cdot u$ is a length-preserving permutation in this case by \autoref{fact:reversibleIsBijection} and we can choose $\bm{p}$ as the pre-image of $\bm{q}$). This yields the cross diagram
      \begin{center}
        \begin{tikzpicture}[baseline=(m-4-1.base)]
          \matrix[matrix of math nodes, text height=1.25ex, text depth=0.25ex] (m) {
                   & u &        & v^\omega &    \\
            \bm{p} &   & \bm{q} &          & {} \\
                   & w &        & \beta    &    \\
          };
          \foreach \j in {1} {
            \foreach \i in {1, 3} {
              \draw[->] let
                \n1 = {int(2+\i)},
                \n2 = {int(1+\j)}
              in
                (m-\n2-\i) -> (m-\n2-\n1);
              \draw[->] let
                \n1 = {int(1+\i)},
                \n2 = {int(2+\j)}
              in
                (m-\j-\n1) -> (m-\n2-\n1);
            };
          };
        \end{tikzpicture}
      \end{center}
      for $w = \bm{p} \circ u$ and, thus, that $w \beta \in Q^* \circ uv^\omega$ is a pre-image of $\beta$ for our function.
    \end{proof}

    In the general case, when the automaton is not reversible, however, the existence of an $\omega$-word with finite orbit does not imply the existence of a periodic $\omega$-word with finite orbit; not even, if the automaton is complete and invertible.
    \begin{prop}\label{prop:groupWithOnlyInfinitePeriodicOrbits}
      Let $\mathcal{T} = (Q, \Sigma, \delta)$ be the \GAut
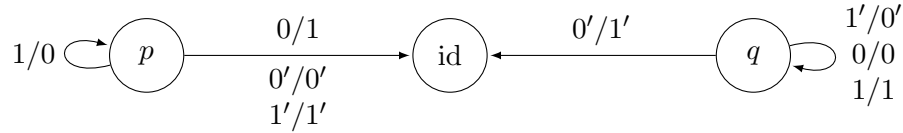
      \begin{center}
        \begin{tikzpicture}[auto, shorten >=1pt, >=latex, node distance=3cm]
          \node[state] (s) {$p$};
          \node[state, right=of s] (e) {$\id$};
          \node[state, right=of e] (t) {$q$};
          
          \draw[->] (s) edge node[above] {$0/1$} node [below, align=center] {$0'/0'$ \\ $1'/1'$} (e)
                    (s) edge[loop left] node {$1/0$} (s)
                    (t) edge[loop right] node[align=center] {$1'/0'$ \\ $0/0$ \\ $1/1$} (t)
                    (t) edge node[swap] {$0'/1'$} (e)
          ;
        \end{tikzpicture}
      \end{center}
      where $\id$ acts as the identity.

      Then, there exist an (ultimately periodic) $\omega$-word $\alpha \in \Sigma^\omega$ with finite orbit $Q^{\pm *} \circ \alpha$ but every periodic $\omega$-word $u^\omega$ with $u \in \Sigma^+$ has an infinite orbit $Q^* \circ u^\omega$.
    \end{prop}
    \begin{proof}
      Let $\alpha = 1' 0^\omega$ and $\alpha' = 0' 0^\omega$. It is easy to see that $p \circ \alpha = \alpha$, $p \circ \alpha' = \alpha'$, $q \circ \alpha = \alpha'$ and $q \circ \alpha' = \alpha$.
      Thus, we have $Q^{\pm *} \circ \alpha = \{ \alpha, \alpha' \}$, which is finite.
      
      To see that the orbit of every periodic word is infinite, let $u \in \Sigma^+$ be arbitrary. We distinguish two cases: $u \in \{ 0, 1 \}^+$ or $u$ contains a $0'$ or a $1'$. For the first case, observe that we obtain the adding machine (see \autoref{ex:addingMachine}) if we remove the state $q$ and the letters $0'$ and $1'$ from $\mathcal{T}$. Thus, for $u \in \{ 0, 1 \}^+$, the orbit of $u^\omega$ is infinite; in fact, we already have $|p^* \circ u^\omega| = \infty$.
      
      In the other case, we can factorize $u = u_0 a_1 u_1 a_2 \dots a_n u_n$ with $a_1, a_2, \dots, a_n \in \{ 0', 1' \}$ and $u_0, u_1, \dots, u_n \in \{ 0, 1 \}^*$. Similarly to the other case, we observe that we obtain the adding machine (with letters $0'$ and $1'$) from $\mathcal{T}$ if we remove the state $p$ and the letters $0$ and $1$. Thus, we have $|q^* \circ (u')^\omega| = \infty$ where $u' = a_1 a_2 \dots a_n$ is obtained by removing all letters in $\{ 0, 1 \}$ from $u$. By the construction of the automaton, reading any of the blocks $u_i \in \{ 0, 1 \}$ does not change the state as long as we are in $q$ or in $\id$. Therefore, it follows easily that $q^* \circ u^\omega$ and, thus, the orbit of $u^\omega$ remain infinite.
    \end{proof}

    \paragraph*{Undecidability of Orbit Finiteness.}
    Algorithmically, it is not possible to decide whether a given $\omega$-word has a finite or infinite orbit. This can be seen from the connection stated in \autoref{thm:torsionIsFiniteDualOrbit}.
    \begin{prop}\label{prop:undecidabilityOfOrbitFiniteness}
      There is some complete and reversible \SAut whose orbit finiteness problem for periodic $\omega$-words
      \problem[
        a complete and reversible \SAut $\mathcal{T} = (Q, \Sigma, \delta)$
      ]{
        a finite word $u \in \Sigma^+$
      }{
        is $Q^* \circ u^\omega$ finite?
      }\noindent
      is undecidable.
    \end{prop}
    \begin{proof}
      There is an automaton group with an undecidable order problem \cite{bartholdi2017wordAndOrderProblems, gillibert2018automaton}, i.\,e.\ there is a \GAut $\mathcal{T}$ such that the problem
      \problem[
        a \GAut $\mathcal{T} = (Q, \Sigma, \delta)$
      ]{
        a state sequence $\bm{q} \in Q^*$
      }{
        is $\bm{q}$ of finite order in $\mathscr{G}(\mathcal{T})$?
      }\noindent
      is undecidable. We reduce this problem to the one in the proposition: as the automaton, we use $\partial \mathcal{T}$, the dual of $\mathcal{T}$, which -- as the dual of a \GAut -- is reversible and complete; the word is $(\rev{\bm{q}})^\omega$, which is periodic. By \autoref{thm:torsionIsFiniteDualOrbit}, the orbit of $(\rev{\bm{q}})^\omega$ under the action of the dual $\partial \mathcal{T}$ is finite if and only if $\bm{q} \circ{}\!$ has torsion (i.\,e.\ is of finite order).
    \end{proof}

    By \autoref{cor:infiniteSemigroupsHaveInfiniteOrbits}, the question whether a given (complete) \SAut admits a word with an infinite orbit or not is equivalent to the finiteness problem for automaton semigroups and, thus, undecidable \cite{Gilbert13}.\footnote{We have already discussed this for the finiteness problem for automaton groups above.} Here, we show a dual result: checking the existence of an $\omega$-word with a finite orbit is undecidable, even for complete and reversible \SAuta.

    \begin{prop}\label{prop:existenceOfWordWithFiniteOrbitIsUndecidable}
      The decision problem
      \problem{
        a complete and reversible \SAut $\mathcal{T} = (Q, \Sigma, \delta)$
      }{
        is there an $\omega$-word $\alpha \in \Sigma^\omega$ such that $|Q^{*} \circ \alpha|<\infty$?
      }
      \noindent{}is undecidable.
    \end{prop}
    \begin{proof}
      By \cite[Theorem~1]{decidabilityPart}, the problem
      \problem{
        a \GAut $\mathcal{T} = (Q, \Sigma, \delta)$
      }{
        is there a state sequence $\bm{q} \in Q^+$ such that $\bm{q} \circ{}\!$ is the identity?
      }
      \noindent{}is undecidable. We reduce this problem to the one in the proposition by taking the dual automaton. Obviously, the dual of a \GAut is complete and reversible.

      Now, suppose that there is some $\bm{q} \in Q^+$ such that $\bm{q} \circ{}\!$ is the identity. Then, $\bm{q} \circ{}\!$, in particular, has torsion. By \autoref{thm:torsionIsFiniteDualOrbit}, this implies that $(\rev{q})^\omega$ has a finite orbit under the action of the dual.

      If, on the other hand, there is some word with a finite orbit under the action of the dual, then, by \autoref{prop:finiteOrbitImpliesPeriodicFiniteOrbit}, this implies that there already is some periodic word $\bm{q}^\omega$ with a finite orbit (where $\bm{q} \in Q^+$). Again, by \autoref{thm:torsionIsFiniteDualOrbit}, this implies that $\rev{\bm{q}} \circ{}\!$ has torsion in the group generated by the original automaton. In other words, there is $k \geq 1$ such that $\rev{\bm{q}}^k \circ{}\!$ is the identity.
    \end{proof}
  \end{section}
  \begin{section}{Reversible but not Bi-Reversible Automata}
    In this section, we have a closer look at the class of automaton groups generated by reversible but not bi-reversible \GAuta.
    
    First, we show that every such \GAut $\mathcal{T}$ admits a periodic word $u^\omega$ with an infinite orbit. The main idea is to take the dual of $\mathcal{T}$, which is a reversible but not bi-reversible \GAut as well, and to find elements without torsion in the semigroup generated by the dual. By \autoref{thm:torsionIsFiniteDualOrbit}, these elements correspond to periodic $\omega$-words with infinite orbits. For the special case that the dual is connected (or only contains non-bi-reversible connected components), we could use \cite[Theorem 23]{GKP} or \cite[Proposition~7]{DaRo16} to obtain that none of the elements have torsion. However, the result is also true in the general case.

    In a reversible and complete \SAut $\mathcal{T} = (Q, \Sigma, \delta)$, all maps ${}\!\cdot u: Q \to Q$ for $u \in \Sigma^*$ are bijections (see \autoref{fact:reversibleIsBijection}). It is not difficult to see that, therefore, in such automata, every connected component is already strongly connected.
    
    The central argument for our proof is that the \emph{semigroup} generated by a reversible but not bi-reversible \GAut cannot contain the (group) inverse of any function induced by a state from a non-bi-reversible connected component.
    \begin{lemma}\label{lem:nonbireversibleComponentDoesNotAllowInversesInSemigroup}
      Let $\mathcal{T} = (Q, \Sigma, \delta)$ be a reversible \GAut with a non-bi-reversible (strongly) connected component consisting of the states $P \subseteq Q$. Then, $\mathscr{S}(\mathcal{T})$ does not contain the inverse $p^{-1} \circ{}\!$ for any $p \in P$.
    \end{lemma}
    \begin{proof}
      We first show that $\mathscr{S}(\mathcal{T})$ must contain ${q}^{-1} \circ{}\!$ for all $q \in P$ if it contains ${p}^{-1} \circ{}\!$ for a single $p \in P$. Therefore, assume the latter to be true. Since $q$ and $p$ are in the same (strongly) connected component, there is some $u \in \Sigma^*$ with $p \cdot u = q$. Then, for $v = p \circ u$, we have ${q}^{-1} \circ{}\! = {p}^{-1} \cdot v \circ{}\!$. Since we have ${p}^{-1} \circ{}\! = \bm{p}' \circ{}\!$ for some $\bm{p}' \in Q^+$ by assumption, we obtain ${q}^{-1} \circ{}\! = \bm{p}' \cdot v \circ{}\!$, which is in $\mathscr{S}(\mathcal{T})$.
      
      Now, assume to the contrary that $\mathscr{S}(\mathcal{T})$ contains ${p}^{-1} \circ{}\!$ for one, and thus for all, $p \in P$. Since $P$ is the state set of some non-bi-reversible component, there are $q, p, r \in P$ and $a, b, c \in \Sigma$ with $p \neq q$ or $a \neq b$ (or both) and the transitions
      
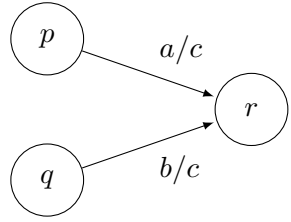
\begin{center}
        \begin{tikzpicture}[auto, shorten >=1pt, >=latex]
          \node[state] (p) {$p$};
          \node[state, below right=0.25cm and 2cm of p] (r) {$r$};
          \node[state, below left=0.25cm and 2cm of r] (q) {$q$};
          
          \draw[->] (p) edge node {$a / c$} (r)
                    (q) edge node[swap] {$b / c$} (r);
        \end{tikzpicture}
      \end{center}
      (i.\,e.\ $p \cdot a = q \cdot b = r$ and $p \circ a = q \circ b = c$). Notice that $p = q$ (which implies $a \neq b$) is not possible since $\mathcal{T}$ is a \GAut and that neither is $a = b$ because $\mathcal{T}$ is reversible.
      
      We have ${p}^{-1} \cdot c \circ{}\! = {r}^{-1} \circ{}\! = {q}^{-1} \cdot c \circ{}\!$ and, by assumption, there are $\bm{p}', \bm{q}' \in Q^+$ with ${p}^{-1} \circ{}\! = \bm{p}' \circ{}\!$ and ${q}^{-1} \circ{}\! = \bm{q}' \circ{}\!$. Together, this yields $\bm{p}' \cdot c \circ{}\! = {r}^{-1} \circ{}\! = \bm{q}' \cdot c \circ{}\!$. Since $\mathcal{T}$ is complete and reversible, the map ${}\! \cdot c: Q^+ \to Q^+$ is a bijection (see \autoref{fact:reversibleIsBijection}) and there is some $k \geq 1$ with $\bm{p}' \cdot c^k = \bm{p}'$ and $\bm{q}' \cdot c^k = \bm{q}'$. For this $k$, we have
      \[
        {p}^{-1} \circ{}\! = \bm{p}' \circ{}\! = (\bm{p}' \cdot c) \cdot c^{k - 1} \circ{}\! = (\bm{q}' \cdot c) \cdot c^{k - 1} \circ{}\! = \bm{q}' \circ{}\! = {q}^{-1} \circ{}\!
      \]
      and, thus, $a = {p}^{-1} \circ c = {q}^{-1} \circ c = b$, which is a contradiction.
    \end{proof}
    
    Actually, \autoref{lem:nonbireversibleComponentDoesNotAllowInversesInSemigroup} allows for a stronger formulation:
    
    \begin{lemma}\label{lem:reversibleNotBireversibleGAutomataArentTorsion}
      Let $\mathcal{T} = (Q, \Sigma, \delta)$ be a reversible \GAut and let $P$ be the non-empty state set of some non-bi-reversible connected component of $\mathcal{T}$. Then, $P$ contains at least two elements and no element $\bm{q} p \circ{}\!$ with $\bm{q} \in Q^+$ and $p \in P$ has an inverse in $\mathscr{S}(\mathcal{T})$. In particular, no element $\bm{q} p \circ{}\!$ has torsion in $\mathscr{G}(\mathcal{T})$.
    \end{lemma}
    \begin{proof}
      Since the connected component belonging to $P$ is non-bi-reversible (but needs to be reversible since $\mathcal{T}$ is), it contains transitions $\trans{q}{a}{b}{p}$ and $\trans{q'}{a'}{b}{p}$ with $a \neq a'$ and $q \neq q'$, as $q = q'$ contradicts the invertibility of $\mathcal{T}$. Therefore, $\{ q, q', p \} \subseteq P$ contains at least two elements.

      Now, let $p \in P$ and $\bm{q} \in Q^*$ be arbitrary and suppose that there is some $\bm{r} \in Q^+$ such that $\bm{r} \bm{q} p \circ{}\!$ is the identity on $\Sigma^*$. Then, $\bm{r} \bm{q} \circ{}\! \in \mathscr{S}(\mathcal{T})$ would be an inverse of $p \circ{}\!$ contradicting \autoref{lem:nonbireversibleComponentDoesNotAllowInversesInSemigroup}.

      Finally, if $\bm{q} p \circ{}\!$ was of torsion, then $(\bm{q} p)^i \circ{}\!$ for some $i$ would be its inverse, which would also constitute a contradiction because of $(\bm{q} p)^i \in Q^+$.
    \end{proof}
    
    We can now apply \autoref{lem:reversibleNotBireversibleGAutomataArentTorsion} to obtain periodic $\omega$-words with infinite orbits.

    \begin{thm}\label{thm:reversibleNotBireversibleGroupsHaveInfiniteOrbitsAtPeriodicPoints}
      Let $\mathcal{T} = (Q, \Sigma, \delta)$ be an \SAut such that its dual $\partial \mathcal{T} = (\Sigma, Q, \partial \delta)$ contains a reversible \GAut $\mathcal{D} = (\Delta, R, \kappa)$ as a sub-automaton.

      Then, every non-bi-reversible connected component of $\mathcal{D}$ with state set $\Gamma$ contains at least two elements and $Q^* \circ u(av)^\omega$ is infinite for all $a \in \Gamma$ and all $u, v \in \Delta^*$.
    \end{thm}
    \begin{proof}
      Let $\Gamma$ be the state set of some non-bi-reversible connected component of $\mathcal{D}$. By \autoref{lem:reversibleNotBireversibleGAutomataArentTorsion}, no element $va \circ{}\!$ from $\mathscr{S}(\mathcal{D})$ with $a \in \Gamma$ and $v \in \Delta^*$ has torsion. By \autoref{thm:torsionIsFiniteDualOrbit}, this means that all orbits $R^* \circ u(a\rev{v})^\omega \subseteq Q^* \circ u(a\rev{v})^\omega$ for $u \in \Delta^*$ are infinite.
    \end{proof}
    
    \begin{cor}\label{cor:reversibleNotBireversibleGroupHasPeriodicWordWithInfiniteOrbit}
      Let $\mathcal{T} = (Q, \Sigma, \delta)$ be a reversible but not bi-reversible \GAut. Then, there are two distinct letters $a, b \in \Sigma$ such that all $Q^* \circ u(av)^\omega$ and all $Q^* \circ u(bv)^\omega$ for $u, v \in \Sigma^*$ are infinite.
    \end{cor}
    \begin{proof}
      Notice that $\partial \mathcal{T}$ is a reversible but not bi-reversible \GAut as well. Thus, the corollary follows from \autoref{thm:reversibleNotBireversibleGroupsHaveInfiniteOrbitsAtPeriodicPoints} since $\partial \mathcal{T}$ must contain at least one non-bi-reversible connected component.
    \end{proof}
    
    Interestingly, we can combine \autoref{thm:reversibleNotBireversibleGroupsHaveInfiniteOrbitsAtPeriodicPoints} with the result about $\omega$-words with a finite orbit from \autoref{prop:finiteOrbitImpliesPeriodicFiniteOrbit} to obtain that many orbits in groups generated by reversible but not bi-reversible \GAuta are \emph{infinite}.
    \begin{cor}\label{cor:infiniteOrbitsInReversibleNonBiReversibleGAutomata}
      Let $\mathcal{T} = (Q, \Sigma, \delta)$ be a reversible but not bi-reversible \GAut and let $\Gamma \subseteq \Sigma$ denote the set of states in the dual automaton $\partial \mathcal{T}$ belonging to a non-bi-reversible connected component.
      
      Then, every $\alpha \in \Sigma^\omega$ which contains at least one letter from $\Gamma$ infinitely often has an infinite orbit: $|Q^* \circ \alpha| = \infty$.
    \end{cor}
    \begin{proof}
      Suppose to the contrary that there is some $\alpha \in \Sigma^\omega$ with a finite orbit such that $\alpha$ contains a letter $a \in \Gamma$ infinitely often. Then, by \autoref{prop:finiteOrbitImpliesPeriodicFiniteOrbit}, there is some $w \in \Sigma^+$ with $w = w_1 a w_2$ for some $w_1, w_2 \in \Sigma^*$ such that the orbit of $w^\omega = w_1 (a w_2 w_1)^\omega$ is finite. However, from \autoref{thm:reversibleNotBireversibleGroupsHaveInfiniteOrbitsAtPeriodicPoints} follows that all words of the form $u (av)^\omega$ must have infinite orbit; a contradiction.
    \end{proof}
        
    The previous corollary directly implies that no infinite word has a finite orbit under the action of a reversible but not bi-reversible \GAut with a connected dual:
    \begin{cor}
      Let $\mathcal{T} = (Q, \Sigma, \delta)$ be a reversible but not bi-reversible \GAut whose dual $\partial \mathcal{T}$ is connected. Then, every $\omega$-word $\alpha \in \Sigma^\omega$ has an infinite orbit $Q^* \circ \alpha$.
    \end{cor}
    \begin{proof}
      Obviously, all letters from $\Sigma$ belong to a non-bi-reversible connected component of the dual. Therefore, any $\alpha \in \Sigma^\omega$ must, in particular, contain at least one of them infinitely often and the result follows from \autoref{cor:infiniteOrbitsInReversibleNonBiReversibleGAutomata}.
    \end{proof}
    
    As a side remark, we note that, in \autoref{lem:reversibleNotBireversibleGAutomataArentTorsion}, we can neither drop the completeness nor the invertibility requirement, which also has consequences for the other results above. To see the former, consider the \SAut $\mathcal{T}$
    {\makeatletter
     \@endparpenalty=10000
     \makeatother
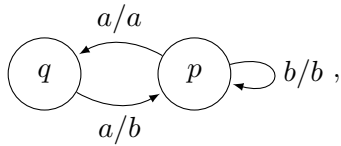
    \begin{center}
      \begin{tikzpicture}[auto, shorten >=1pt, >=latex, baseline=(p.base)]
        \node[state] (q) {$q$};
        \node[state, right=of q] (p) {$p$};
        \draw[->] (q) edge[bend right] node[swap] {$a/b$} (p)
                  (p) edge[bend right] node[swap] {$a/a$} (q)
                      edge[loop right] node {$b/b$} (p)
        ;
      \end{tikzpicture},
    \end{center}}\noindent
    which is strongly connected, reversible and invertible but neither bi-reversible nor complete. One may observe that $q^2 \circ{}\!$ is undefined on all words (except $\varepsilon$) and that $q \circ{}\!$, therefore, has torsion as we have $q^2 \circ{}\! = q^3 \circ{}\!$. In fact, it turns out that the semigroup generated by the automaton is finite\footnote{This can, for example, be seen by computing its powers (up to the fourth one); the left Cayley graph of the generated semigroup can be found in \autoref{fig:CayleyGraphOfT1}.} and, thus, that all its elements have torsion.
    \begin{figure}
      \begin{center}
        \begin{tikzpicture}[auto, shorten >=1pt, >=latex]
          \node[state] (p) {$p \circ{}\!$};
          \node[state, below=of p] (q) {$q \circ{}\!$};
          \node[state, right=of p] (qp) {$qp \circ{}\!$};
          \node[state, above=of qp] (pp) {$p^2 \circ{}\!$};
          \node[state, right=of q] (pq) {$pq \circ{}\!$};
          \node[state, right=of pq] (qq) {$q^2 \circ{}\!$};
          \node[state, right=of pp] (qpp) {$qp^2 \circ{}\!$};
          
          \draw[->] (p) edge node {$p \circ{}\!$} (pp)
                        edge node {$q \circ{}\!$} (qp)
                    (q) edge node {$p \circ{}\!$} (pq)
                        edge[bend right] node[swap] {$q \circ{}\!$} (qq)
                    (pp) edge[loop left] node {$p \circ{}\!$} (pp)
                         edge node {$q \circ{}\!$} (qpp)
                    (qp) edge[loop right] node {$p \circ{}\!$} (qp)
                         edge node {$q \circ{}\!$} (qq)
                    (pq) edge[out=130,in=100,looseness=8] node {$p \circ{}\!$} (pq)
                         edge node {$q \circ{}\!$} (qq)
                    (qq) edge[loop right] node {$p \circ{}\!, q \circ{}\!$} (qq)
                    (qpp) edge[loop right] node {$p \circ{}\!$} (qpp)
                          edge[bend left] node {$q \circ{}\!$} (qq)
          ;
        \end{tikzpicture}
      \end{center}\vspace{-\baselineskip}
      \caption{Left Cayley graph of $\mathscr{S}(\mathcal{T}_1)$.}\label{fig:CayleyGraphOfT1}
    \end{figure}
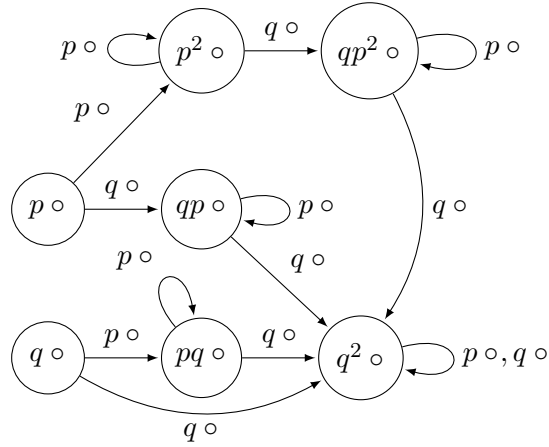

    To see that the automaton in the statements of \autoref{lem:reversibleNotBireversibleGAutomataArentTorsion} and the corollaries needs to be invertible, we use the following counter-example based on the connection from \autoref{thm:torsionIsFiniteDualOrbit}.
    \begin{counterexample}\label{prop:infiniteSemigroupWithAllPeriodicOrbitsFinite}
      Let $\mathcal{T} = (Q, \Sigma, \delta)$ be the \GAut
      \begin{center}
        \begin{tikzpicture}[auto, shorten >=1pt, >=latex, baseline=(b.base)]
          \node[state] (b) {$b$};
          \node[state, above right=of b] (a) {$a$};
          \node[state, below right=of b] (d) {$d$};
          \node[state, below right=of a] (c) {$c$};
          \node[state, right=of c] (id) {$\textnormal{id}$};
          
          \draw[->] (a) edge[bend left] node[align=center] {$0/1$\\$1/0$} (id)
                    (b) edge node {$0/0$} (a)
                    (b) edge node[swap] {$1/1$} (c)
                    (c) edge node {$0/0$} (a)
                    (c) edge node {$1/1$} (d)
                    (d) edge[bend right] node[swap] {$0/0$} (id)
                    (d) edge node {$1/1$} (b)
                    (id) edge[loop right] node[align=center] {$0/0$\\$1/1$} (id)
                    ;
        \end{tikzpicture},
      \end{center}
      which generates the Grigorchuk group. The Grigorchuk group $\mathscr{G}(\mathcal{T})$ is infinite (see e.\,g.\ \cite[Theorem~1.6.1]{nekrashevych2005self}) and, thus, so is $\mathscr{S}(\mathcal{T})$ (see, for example, \cite[Proof of Lemma~5]{aklmp12})\footnote{In fact, we have $\mathscr{S}(\mathcal{T}) = \mathscr{G}(\mathcal{T})$ since the group is a Burnside group (see e.\,g.\ \cite[Theorem~1.6.1]{nekrashevych2005self}).}.
      
      The dual $\partial \mathcal{T}$ of $\mathcal{T}$\enlargethispage{\baselineskip}
      {\makeatletter
       \@endparpenalty=10000
       \makeatother
      \begin{center}
        \begin{tikzpicture}[auto, shorten >=1pt, >=latex, baseline=(b.base)]
          \node[state] (0) {$0$};
          \node[state, right=of 0] (1) {$1$};
          
          \draw[->] (0) edge[loop left] node[align=center] {$\textnormal{id}/\textnormal{id}$\\$b/a$\\$c/a$\\$d/\textnormal{id}$} (0)
                    (0) edge[bend left] node {$a/\textnormal{id}$} (1)
                    (1) edge[loop right] node[align=center] {$\textnormal{id}/\textnormal{id}$\\$b/c$\\$c/d$\\$d/b$} (1)
                    (1) edge[bend left] node {$a/\textnormal{id}$} (0);
        \end{tikzpicture}
      \end{center}}\noindent
      is complete and reversible but neither invertible nor bi-reversible. It generates an infinite semigroup $\mathscr{S}(\partial \mathcal{T})$ since an \SAut generates a finite semigroup if and only if its dual does \cite[Proposition~10]{aklmp12}. However, the orbits of all ultimately periodic words $\Sigma^* \circ_{\partial} \bm{p}\bm{q}^\omega$ under its action with $\bm{p} \in Q^*$ and $\bm{q} \in Q^+$ are finite.
      
      This is the case because all elements of the Grigorchuk group $\mathscr{G}(\mathcal{T})$ have finite order (see e.\,g.\ \cite[Theorem~1.6.1]{nekrashevych2005self}), i.\,e.\ $\mathscr{G}(\mathcal{T})$ is a Burnside group, and because \autoref{thm:torsionIsFiniteDualOrbit} yields that the orbit $\Sigma^* \circ_{\partial} \bm{p} \bm{q}^\omega$ for any $\bm{p} \in Q^*$ and $\bm{q} \in Q^+$ is finite if and only if $\rev{\bm{q}} \circ{}\!$ has torsion in $\mathscr{S}(\mathcal{T})$.
    \end{counterexample}
    
    Another important consequence of \autoref{prop:infiniteSemigroupWithAllPeriodicOrbitsFinite} is that an infinite automaton semigroup does not always admit a periodic or ultimately periodic word with an infinite orbit. Thus, \autoref{cor:infiniteSemigroupsHaveInfiniteOrbits} cannot be extended in this way. However, since the automaton $\partial \mathcal{T}$ is not a \GAut, this still leaves the following question for automaton groups open.
    \begin{prob}
      Does every \GAut generating an infinite group admit a periodic or ultimately periodic $\omega$-word with an infinite orbit?
    \end{prob}

    In fact, for arbitrary \SAuta $\mathcal{T}$ and $\mathcal{T}'$, an isomorphism between $\mathscr{S}(\mathcal{T})$ and $\mathscr{S}(\mathcal{T}')$ \emph{does not} imply that $\mathscr{S}(\partial \mathcal{T})$ and $\mathscr{S}(\partial \mathcal{T}')$ are isomorphic. In other words, the dual is not an algebraic property of an automaton semigroup but only a property of the presentation by a specific automaton. The dual of the Grigorchuk automaton depicted in \autoref{prop:infiniteSemigroupWithAllPeriodicOrbitsFinite} generates the free semigroup of rank two\footnote{This follows, for example, from the fact that reversible two-state \SAuta generate either free of finite semigroups \cite{Klimann2016}.}. This semigroup can also be generated by a different automaton (see \cite[Proposition~4.1]{cain2009automaton}) which admits a periodic $\omega$-word with an infinite orbit (in fact, every $\omega$-word has an infinite orbit under its action).

    Therefore, we have not completely settled the semigroup case either. It could still be the case that every infinite automaton semigroup is generated by some \SAut admitting a periodic word with infinite orbit.
  \end{section}

\bibliographystyle{plain}
\bibliography{references}

\begin{thebibliography}{10}

\bibitem{aklmp12}
Ali Akhavi, Ines Klimann, Sylvain Lombardy, Jean Mairesse, and Matthieu
  Picantin.
\newblock On the finiteness problem for automaton (semi)groups.
\newblock {\em International Journal of Algebra and Computation}, 22(06):1--26,
  2012.

\bibitem{bartholdi2017wordAndOrderProblems}
Laurent Bartholdi and Ivan Mitrofanov.
\newblock The word and order problems for self-similar and automata groups.
\newblock {\em Groups, Geometry, and Dynamics}, 14:705--728, 2020.

\bibitem{brough2015automaton}
Tara Brough and Alan~J. Cain.
\newblock Automaton semigroup constructions.
\newblock {\em Semigroup Forum}, 90(3):763--774, 2015.

\bibitem{brough2017automatonTCS}
Tara Brough and Alan~J. Cain.
\newblock Automaton semigroups: new constructions results and examples of
  non-automaton semigroups.
\newblock {\em Theoretical Computer Science}, 674:1--15, 2017.

\bibitem{cain2009automaton}
Alan~J. Cain.
\newblock Automaton semigroups.
\newblock {\em Theoretical Computer Science}, 410(47):5022--5038, 2009.

\bibitem{orbitsPart}
Daniele D'Angeli, Dominik Francoeur, Emanuele Rodaro, and Jan~Philipp Wächter.
\newblock Infinite automaton semigroups and groups have infinite orbits.
\newblock {\em Journal of Algebra}, 553:119 -- 137, 2020.

\bibitem{DaRo16}
Daniele D'Angeli and Emanuele Rodaro.
\newblock Freeness of automaton groups vs boundary dynamics.
\newblock {\em Journal of Algebra}, 462:115--136, 2016.

\bibitem{decidabilityPart}
Daniele D'Angeli, Emanuele Rodaro, and Jan~Philipp W{\"a}chter.
\newblock Automaton semigroups and groups: on the undecidability of problems
  related to freeness and finiteness.
\newblock {\em Israel Journal of Mathematics}, 2020.

\bibitem{expandabilityPart}
Daniele D'Angeli, Emanuele Rodaro, and Jan~Philipp W{\"a}chter.
\newblock Orbit expandability of automaton semigroups and groups.
\newblock {\em Theoretical Computer Science}, 809:418 -- 429, 2020.

\bibitem{DAngeli2017}
Daniele D'Angeli, Emanuele Rodaro, and Jan~Philipp Wächter.
\newblock On the complexity of the word problem for automaton semigroups and
  automaton groups.
\newblock {\em Advances in Applied Mathematics}, 90:160 -- 187, 2017.

\bibitem{structurePart}
Daniele D'Angeli, Emanuele Rodaro, and Jan~Philipp Wächter.
\newblock On the structure theory of partial automaton semigroups.
\newblock {\em Semigroup Forum}, 2020.

\bibitem{Gilbert13}
Pierre Gillibert.
\newblock The finiteness problem for automaton semigroups is undecidable.
\newblock {\em International Journal of Algebra and Computation}, 24(01):1--9,
  2014.

\bibitem{gillibert2018automaton}
Pierre Gillibert.
\newblock An automaton group with undecidable order and {{E}}ngel problems.
\newblock {\em Journal of Algebra}, 497:363 -- 392, 2018.

\bibitem{GKP}
Thibault Godin, Ines Klimann, and Matthieu Picantin.
\newblock On torsion-free semigroups generated by invertible reversible {M}ealy
  automata.
\newblock In Adrian-Horia Dediu, Enrico Formenti, Carlos Mart{\'i}n-Vide, and
  Bianca Truthe, editors, {\em Language and Automata Theory and Applications
  (LATA 2015)}, volume 8977 of {\em LNCS}, pages 328--339. Springer, Berlin,
  2015.
\newblock Arxiv:1410.4488.

\bibitem{GriNeShu}
Rostislav~I. Grigorchuk, Volodymyr~V. Nekrashevych, and Vitaly~I. Sushchansky.
\newblock Automata, dynamical systems, and groups.
\newblock {\em Proc. Steklov Inst. Math.}, 231(4):128--203, 2000.

\bibitem{grigorchuk2008groups}
Rostislav~I. Grigorchuk and Igor Pak.
\newblock Groups of intermediate growth: an introduction.
\newblock {\em L’Enseignement Mathématique}, 54:251--272, 2008.

\bibitem{Klimann2016}
Ines Klimann.
\newblock Automaton semigroups: The two-state case.
\newblock {\em Theory of Computing Systems}, 58(4):664--680, May 2016.

\bibitem{KPS}
Ines Klimann, Matthieu Picantin, and Dmytro Savchuk.
\newblock A connected 3-state reversible {M}ealy automaton cannot generate an
  infinite {B}urnside group.
\newblock {\em International Journal of Foundations of Computer Science},
  29(2):297--314, 2018.

\bibitem{nekrashevych2005self}
Volodymyr~V. Nekrashevych.
\newblock {\em Self-similar groups}, volume 117 of {\em Mathematical Surveys
  and Monographs}.
\newblock American Mathematical Society, Providence, RI, 2005.

\bibitem{Su-Ve09}
Zoran {\v{S}}uni{\'c} and Enric Ventura.
\newblock The conjugacy problem in automaton groups is not solvable.
\newblock {\em Journal of Algebra}, 364:148--154, 2012.

\bibitem{pspacePart}
Jan~Philipp Wächter and Armin Weiß.
\newblock An automaton group with {$\PSpace$}-complete word problem.
\newblock In Christophe Paul and Markus Bläser, editors, {\em 37th
  International Symposium on Theoretical Aspects of Computer Science (STACS)},
  volume 154 of {\em LIPIcs}, pages 6:1--6:17. Schloss Dagstuhl -
  Leibniz-Zentrum für Informatik, 2020.

\end{thebibliography}

\end{document}